\documentclass[10pt]{article}
\usepackage{geometry}
\usepackage[utf8]{inputenc}
\usepackage{amsmath}
\usepackage{amssymb}
\usepackage{amsthm}
\usepackage{hyperref}
\usepackage{cleveref}
\usepackage{graphicx}
\usepackage{bm}
\usepackage{mathtools}
\usepackage{subfig}
\usepackage[font=small,labelfont=bf]{caption}

\newtheorem{proposition}{Proposition}
\newtheorem{theorem}{Theorem}
\newtheorem{corollary}[proposition]{Corollary}

\newtheorem{lemma}[proposition]{Lemma}

\newtheorem{example}{Example}
\newtheorem{definition}{Definition}

\newcommand{\EE}{\mathbb{E}}
\newcommand{\RR
}{\mathbb{R}}
\newcommand{\DD}{\mathcal{D}}

\renewcommand{\epsilon}{\varepsilon}
\newcommand{\simiid}{\overset{\mathrm{iid}}{\sim}}

\DeclareMathOperator{\mono}{mono}
\DeclareMathOperator{\poly}{poly}
\DeclareMathOperator{\Rank}{Rank}

\DeclareMathOperator{\Var}{Var}
\DeclareMathOperator{\Beta}{beta}

\usepackage{xcolor}
\newcommand{\Comments}{1}
\newcommand{\mynote}[3]{\ifnum\Comments=1\textcolor{#1}{#2: #3}\fi}

\newcommand{\colleges}{\mathcal{F}}
\newcommand{\numcolleges}{m}
\newcommand{\totalsupply}{S}
\newcommand{\college}{f}

\usepackage[round]{natbib}
\hypersetup{
    colorlinks,
    linkcolor={red!50!black},
    citecolor={blue!50!black},
    urlcolor={blue!80!black}
}

\title{Monoculture in Matching Markets}
\author{Kenny Peng\\Nikhil Garg}
\date{\url{kennypeng@cs.cornell.edu}, \url{ngarg@cornell.edu}}

\begin{document}

\maketitle

\begin{abstract}
    Algorithmic monoculture arises when many decision-makers rely on the same algorithm to evaluate applicants. An emerging body of work investigates possible harms of this kind of homogeneity, but has been limited by the challenge of incorporating market effects in which the preferences and behavior of many applicants and decision-makers jointly interact to determine outcomes.
    
    Addressing this challenge, we introduce a tractable theoretical model of algorithmic monoculture in a two-sided matching market with many participants. We use the model to analyze outcomes under monoculture (when decision-makers all evaluate applicants using a common algorithm) and under polyculture (when decision-makers evaluate applicants independently). All else equal, monoculture (1) selects less-preferred applicants when noise is well-behaved, (2) matches more applicants to their top choice, though individual applicants may be worse off depending on their value to decision-makers and risk tolerance, and (3) is more robust to disparities in the number of applications submitted.
\end{abstract}

\section{Introduction}
What happens when many decision-makers---in consequential domains like employment and college admissions---all evaluate applicants using the \textit{same} algorithm? This possibility has been termed ``algorithmic monoculture''---a name that references the problems arising in agriculture when a single crop is farmed extensively, amplifying the risk of correlated widespread failures. A recent line of work considers how a lack of variety in algorithms may warrant similar concerns.

Initial inquiries into the consequences of algorithmic monoculture have identified potential harms to both decision-makers and applicants. In seminal work initiating the formal study of algorithmic monoculture, \cite{kleinberg2021algorithmic} demonstrated that firms relying on a common algorithm may hire weaker applicants than when using idiosyncratic (but less individually accurate) hiring processes. On the applicant side, an emerging set of work highlights the concern of \textit{systemic exclusion}, which occurs when an applicant is denied from all opportunities \citep{creel2022algorithmic, bommasani2022picking, toups2023ecosystem, jain2023algorithmic}.

This emerging literature has been limited by the challenge of incorporating \textit{market} effects. Domains where monoculture may be especially impactful---such as employment and college admissions---share some common properties. First, participants on both sides of the market are in competition: workers compete over a limited number of job openings, and firms compete over workers who can each accept only one offer. Second, outcomes are determined by the preferences of participants on both sides: roughly speaking, a match is formed only if both a worker and firm prefer each other over their other options. Analyzing these interactions across many participants can be theoretically challenging; indeed, the existing literature considers neither markets with many participants nor two-sided interactions.

The technical contribution of our paper is a matching markets model that can be used to study algorithmic monoculture. In particular, we analyze equilibrium outcomes (i.e., stable matchings) that emerge under \textit{monoculture} (when firms each use the same method to evaluate applicants) and \textit{polyculture} (when firms use idiosyncratic methods to evaluate applicants). The model builds directly on the popular continuum model formalized by \cite{azevedo2016supply}, in which stable matchings are characterized by a cutoff structure in which applicants are matched to their favorite firm among those they can ``afford''---i.e., among the firms at which the applicant's \textit{score} exceeds the firm's \textit{cutoff}. In the model we introduce, an applicant's score is equal to an \textit{estimated value} that is a noisy realization of their \textit{true value}. Under monoculture, all firms use a single shared estimate of an applicant's value; under polyculture, firms obtain separate, independently drawn estimates of the applicant's value.

We show how modeling market effects with many participants advances the present understanding of monoculture in several ways, with \Cref{fig:dotsaveragerankmatched} illustrating our first two theoretical results. 
\begin{itemize}
    \item First, as long as noise is well-behaved (i.e., with lighter-than-exponential tail), essentially only the most preferred applicants are hired under polyculture when there are many firms. This shows that the existing (relatively modest) effect found by \cite{kleinberg2021algorithmic} with two firms can be fully strengthened when considering many firms. Our result gives a ``wisdom of the crowds''-like result, where even when individual firms are noisy, the market outcome can behave as if firms each had perfect information about applicants---even despite firms not explicitly sharing any information. Polyculture can thus significantly outperform monoculture in terms of firm welfare, even when the idiosyncratic evaluation processes of firms are much noisier than the shared algorithm used under monoculture.
    \item Second, we challenge---and add nuance to---existing work suggesting that monoculture harms applicants. In fact, we show that \textit{on average}, applicants are better off under monoculture than polyculture, in the sense that on average, applicants are matched to firms that they more prefer.\footnote{This finding is closely related to work in school choice that compares tie-breaking rules, where it has been shown that the use of a single lottery number, rather than multiple lottery numbers, results in more students being matched to their top choices (e.g., \cite{abdulkadirouglu2009strategy, de2023performance, ashlagi2019assigning, arnosti2023lottery, allman2023rank}). We discuss this connection further in \Cref{sec:extended_related}.} On the other hand, the \textit{individual} preference of an applicant for monoculture or polyculture varies based on their true value. Some applicants are strictly more likely to be matched to their top choice under monoculture, but strictly less likely to match at all; for these applicants, their preference could depend on other factors such as risk aversion.
    
    \item Third, we expand the scope of existing analysis by extending our model to incorporate \textit{differential application access}: when applicants submit a varying number of applications. This addresses an equity concern especially salient in college admissions, where the cost of application can be high.\footnote{Application fees to U.S. colleges were as high as 100 dollars in the 2022-2023 cycle \citep{wood2022colleges}. Moreover, the Common Application reported that the number of applications submitted by students varied significantly in the 2021-2022 cycle; for example, high-volume applicants who submitted 15 or more applications were 2.5x as likely to have attended a private high school \citep{kim2022firstyear}.} It also reveals how different application systems incentivize applicants to submit more applications. We find that applicants who submit many applications benefit much more under polyculture than under monoculture. Consequently, differences in the number of applications submitted harms college and firm welfare more under polyculture than monoculture. Overall, monoculture is more robust than polyculture with respect to differential application access.
\end{itemize} 
In this way, our approach strengthens, challenges, and broadens the scope of the existing monoculture literature. The effects we study emerge in large markets with two-sided preferences, the standard setting for algorithmic monoculture (e.g., prior works often motivate monoculture and polyculture through the lens of hiring). Moreover, the cutoff structure inherent in \cite{azevedo2016supply} results in clean intuition for each of our results---as well as precise characterizations at the level of individual applicants. Our work further contributes back to the matching literature; in particular, we study participant welfare with respect to their \textit{true preferences}, but in markets where matches are made according to \textit{noisy preferences}. In this way, our results shed light on stable matching outcomes under imperfect information.

\begin{figure}
    \begin{center}
    \includegraphics[width=\linewidth]{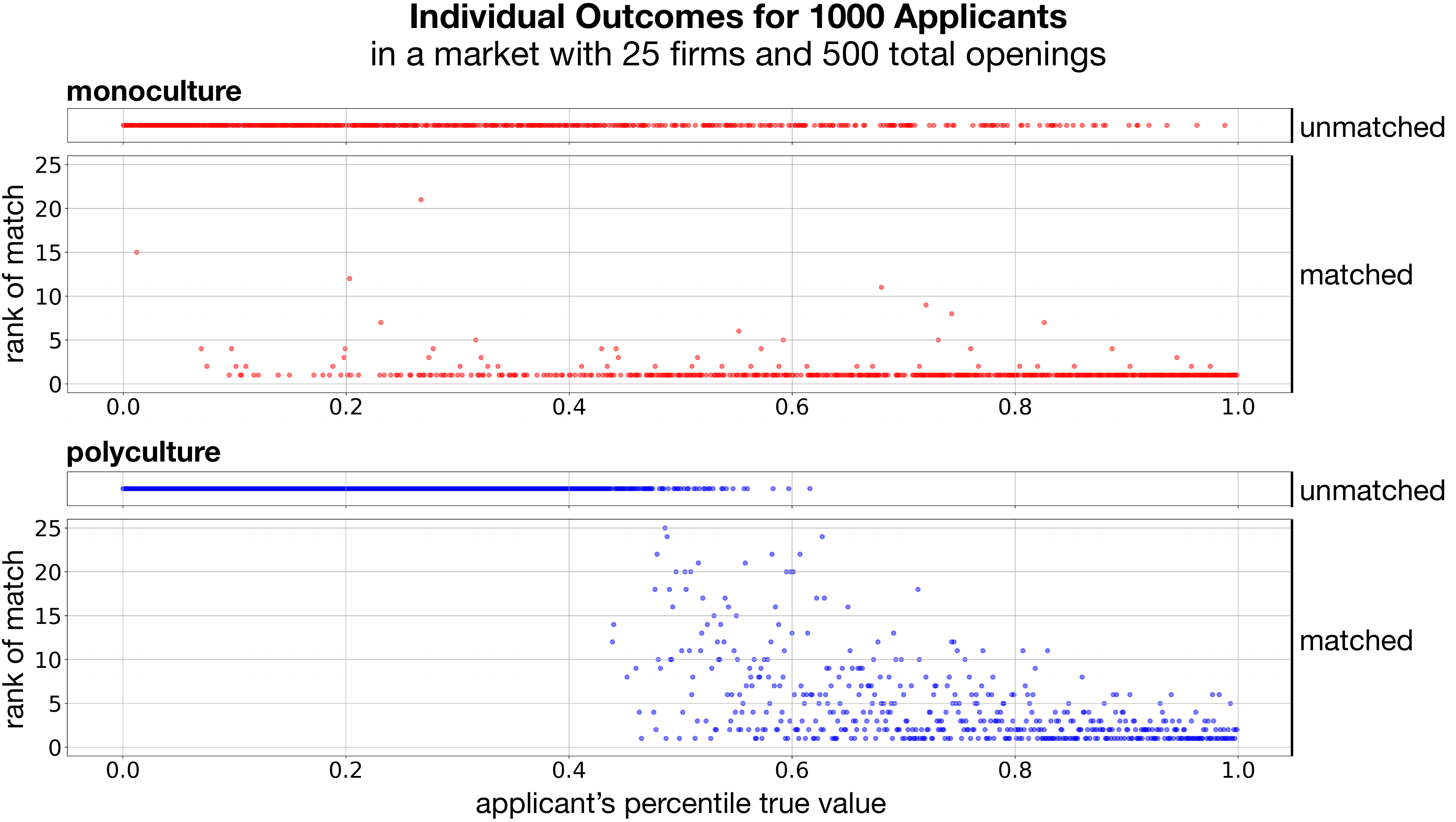}
    \caption[LoF entry]{Under monoculture, firm welfare is lower since there is more noise in who is matched (both low- and high-value applicants are matched), whereas under polyculture, almost exclusively high-value applicants are matched. However, overall applicant welfare is higher under monoculture, since the applicants who are matched tend to be matched with their top choices, whereas under polyculture, applicants match more frequently to lower-ranked choices. Effects further vary conditional on an applicant's true value.}\label{fig:dotsaveragerankmatched}
    \end{center}
\end{figure}

We now provide a more detailed overview of our model, results, and computational experiments.

\subsection{Overview of Model}

\begin{figure}[tbh]
    \begin{center}
        \includegraphics[width=0.9\linewidth]{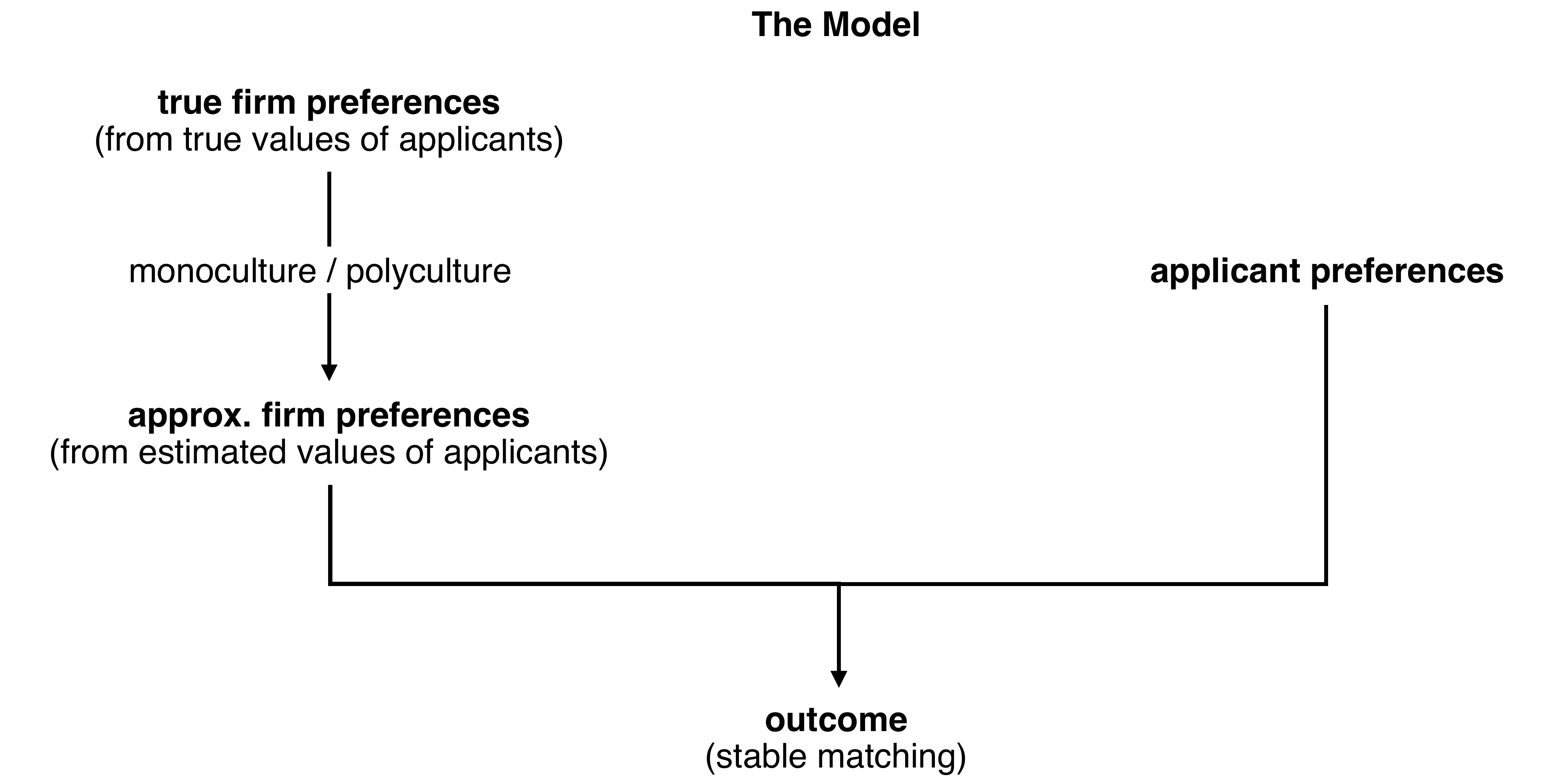}
    \end{center}
    \caption{Overview of our model. True preferences of firms depend on the \textit{value} of each applicant (firms all prefer higher values). However, firms do not know the true value of an applicant, but instead, each obtains a noisy \textit{estimated value}, which determine their approximate preferences. We model the market outcome as a stable matching with respect to the approximate preferences of firms and the true preferences of applicants.}
\end{figure}
We provide an informal overview of our theoretical model. In our model, each applicant has a single \textit{true value}, and firms each prefer higher-value applicants. However, firms do not have access to the true value of each applicant; rather they each obtain an \textit{estimated value} for the applicant. In our model, the estimated value is equal to the true value perturbed by some random noise. Each firm forms an approximation of their preferences according to their estimated values of applicants. Then:
\begin{itemize}
    \item In \textbf{monoculture}, all firms have a shared estimated value for each applicant (obtained, for example, through a shared algorithm).
    \item In \textbf{polyculture}, each firm has independent estimated values for each applicant (obtained, for example, through idiosyncratic evaluation processes).
\end{itemize}
So under monoculture, the firms all have the same approximate preferences over applicants (i.e., ranked list).

Applicants also have preferences over firms, which we assume reflect their true preferences. In our theoretical analysis, applicant preferences are set uniformly at random.\footnote{This is a common assumption to allow for tractability: e.g., \cite{ashlagi2017unbalanced}. In computational experiments, we show that the directional predictions of our results hold given more general (correlated) preferences.} Firms each have the same capacity and the total capacity of firms is less than the total number of applicants, so some applicants go unmatched.

\paragraph{Solution concept: Stable Matchings.} From the preferences of applicants and firms (which for firms are approximate), we study the outcome of the market, which we model as a \textit{stable matching} \citep{gale1962college}. A \textit{matching} assigns a subset of applicants to firms, such that firms are not matched with more applicants than they have the capacity for. Applicants can go unmatched. A matching is \textit{stable} if no pair of applicant and firm would jointly prefer to defect from the current matching to be matched with each other.

Stable matchings have a clean characterization in terms of \textit{cutoffs}, which can be described as follows. Each firm sets a cutoff. An applicant is given an ``offer'' by the firm if their estimated value at the firm exceeds the cutoff. An applicant then ``accepts'' their most preferred offer, and is matched to that firm. If these cutoffs are \textit{market-clearing}---roughly meaning that the capacities of firms are filled but not exceeded---then the associated matching is stable. Moreover, all stable matchings can be characterized by market-clearing cutoffs in this way \citep{azevedo2016supply}.

\paragraph{Formal setup and some key observations.}
In the formal model we consider---which builds on the model of \cite{azevedo2016supply}---there is a continuum of applicants. This means that there is no stochasticity in the model, even when applicants have randomly drawn preferences and estimated values. Given the symmetry of our theoretical setup (applicants have uniformly random preferences over firms), market-clearing cutoffs are the \textit{same} for all firms in the continuum model; each firm has the same ``bar'' for acceptance. This allows for tractable analysis of stable matchings under both monoculture and polyculture. In both cases, an applicant is matched (i.e., extended at least one offer) if and only if they have at least one estimated value that exceeds this common cutoff. In monoculture, each applicant has a single estimated value used by each firm, so they are matched if and only if this single estimated value exceeds the cutoff. Under polyculture, each applicant has a different estimated value used by each firm, so they are matched if and only if their maximum estimated value exceeds the cutoff. These two basic observations form the central pillars of our analysis---and underpin stark contrasts between market outcomes under monoculture and polyculture. %

\paragraph{Where is the model applicable?}
Our model is intended to describe many-to-one matching markets that share a few key features. First, one side of the market has shared preferences over the other side. Our canonical examples are firms or colleges that have shared preferences over applicants. Therefore, our model is likely most directly applicable for firms in the same industry. Meanwhile, we expect that many colleges have similar preferences over students. Considering the effect of monoculture in the presence of more heterogeneous preferences would be interesting future work---one might expect that monoculture would harm welfare in this setting since applicants' idiosyncratic firm-specific strengths would receive less weight.

A second feature of our model is that one side of the market \textit{approximates} their preferences over the other side. In our model, the firm/college side engages in this process. Indeed, in practice, firms and colleges do not know their true preferences (i.e., their preferences given perfect information about applicants), but rather rely on tools such as resumes, interviews, essays, or test estimated values to estimate their preferences. Future work might consider the effect of preference approximation on the other side (or on both sides) of the market.

Finally, our model assumes that both sides of the market prefer to be matched over being unmatched. For example, colleges prefer to fill seats over leaving seats empty. In particular, this means that the number of applicants who are matched remains is fixed in our model. One might consider what would happen when participants have stronger preferences, or consider downstream consequences, such as the effect of matching quality on the number of future available positions.

\subsection{Overview of Theoretical Results and Comparison to Past Work}

We now summarize the main theoretical results we obtain in our model and discuss the intuition behind these results. We show in computational experiments that these results provide accurate predictions of the outcomes of monoculture relative to the outcomes of polyculture in a wide range of simulated matching markets. 

We begin by focusing on two broad questions: how does monoculture compare to polyculture in terms of (1) the welfare of firms, and (2) the welfare of applicants? We then (3) extend our model to consider inequity, in terms of how many applications each applicant is able to submit.

\paragraph{Firm welfare.}
Our first main result (\Cref{thm:wisdom}) shows that under polyculture, the probability an applicant with a given value is matched approaches either $0$ or $1$ as the number of firms grows large. In other words, the set of applicants who are matched is essentially deterministic: those with values above some threshold are almost certainly matched, and those with values below the threshold are almost never matched. Therefore, under polyculture, firms achieve a ``wisdom of the crowds''-like effect, where they collectively match with only the highest-value applicants. This may be surprising given that each firm still only ranks applicants based on an individual noisy ranking of applicants; no information is explicitly shared across firms, and each firm's noise can be arbitrarily high. A consequence of the result is that firm welfare approaches optimal under polyculture when there are many firms, and therefore is higher than under monoculture.

We give some intuition for this result. Under polyculture, we show that whether a student is matched depends on the \textit{maximum} estimated value they obtain across firms---if the applicant has just one estimated value that exceeds the cutoff, they are matched. (Recall that these estimated values are independent noisy perturbations of an applicant's value.) The result then relies on a property of the distribution of an applicant's maximum estimated value. When the noise distribution has a lighter-than-exponential tail (e.g., uniform, Gaussian), this distribution \textit{concentrates}. Then, given two applicants, one with a higher value than the other, the probability that the maximum estimated value of the higher-value applicant exceeds the maximum estimated value of the lower-value applicant approaches $1$ as the number of estimated values (i.e., the number of firms) grows large. In other words, the maximum estimated value is able to effectively distinguish between students of different values, even when a single estimated value cannot. To give a concrete analogy, suppose you were trying to assess the ability of several students on a particular topic and administered a series of tests. Each individual test gives an independent but noisy assessment of a student's competence; therefore, some students with lower ability may outperform students of higher ability on a given test. Clearly, taking the average performance of students would distinguish their relative abilities given enough tests. Here, we use the fact that the \textit{best} performance of each student can also provide an accurate ordering of the students' abilities.

\Cref{thm:wisdom} may be compared to the findings of \cite{kleinberg2021algorithmic}, who consider a similar theoretical setup in a setting with two firms that each select one applicant. They find the same directional impact of polyculture: when two firms independently evaluate applicants, they can obtain better outcomes than when using a shared evaluation---even when this shared evaluation is more accurate than each individual independent evaluation. By considering a large matching market, with many firms and applicants,
we show that the effect of independent evaluations is significantly amplified when there are more firms. \Cref{thm:wisdom} shows that when there are many firms, not only does polyculture outperform monoculture, but polyculture can approximate the \textit{optimal} matching for firms where only the highest value applicants are chosen. Our model also clarifies the reason why polyculture can have this effect: the \textit{maximum} of individually noisy evaluations of an applicant can effectively distinguish between applicants.

\paragraph{Applicant Welfare.}
Our second set of results pertains to the welfare of applicants, which depends on if they (1) are matched at all, and (2) are matched to firms (colleges) that they more prefer. \Cref{thm:top-choice} considers both these desiderata. 

First, we show that all applicants are at least as likely to be matched to their \textit{top choice} under monoculture in comparison to polyculture---and in fact, that all applicants who are matched to some firm under monoculture are matched to their top choice. Intuitively, this is because under monoculture, any applicant who receives a high estimated value at one firm receives the same high estimated value at all firms, so applicants who receive offers tend to receive many offers. Therefore, monoculture maximizes total applicant welfare, as defined by the average rank of matches from the applicant perspective.

For individual applicants, however, monoculture is not always strictly better. For example, while all applicants are at least as likely to be matched to their top choice under monoculture, some applicants are less likely to be matched at all. Applicants above a threshold are almost certainly matched under polyculture; under monoculture, they may have a substantial probability of remaining unmatched. For a large set of applicants, there is a genuine trade-off between monoculture and polyculture, in the sense that under monoculture these applicants are strictly more likely to be matched to their top choice while strictly less likely to be matched overall. Thus, neither monoculture nor polyculture \textit{stochastically dominates} the other, and applicant preference for monoculture or polyculture depends, for example, on their level of risk aversion. \Cref{fig:dotsaveragerankmatched} illustrates these effects, showing average match rank and probability of matching for different applicants.  

The existing literature concerning the effect of algorithmic monoculture on applicant welfare has focused heavily on the concept of \textit{systemic exclusion}, where an applicant is denied opportunity from all firms \citep{creel2022algorithmic}. In the context of algorithmic classifiers, \cite{bommasani2022picking} study this by defining and measuring \textit{outcome homogeneity}. Our results challenge the prevalent intuition that monoculture increases the rate of systemic exclusion, i.e., the probability of an applicant being unmatched. Indeed, in our model, the total number of applicants matched remains the same regardless of whether applicants are evaluated according to monoculture or polyculture. This reflects settings in which the total capacity is constant---when colleges have a fixed number of seats, and when firms are looking to fill a fixed number of positions. Our results instead establish that the distribution of \textit{who} is systematically excluded (unmatched) changes: compared to polyculture, monoculture is more likely to exclude students who would be matched if firms had perfect information. Furthermore, the students who are more likely to be excluded under monoculture are sometimes more likely to match with their \textit{first} preference under monoculture. By taking a market-level view, the relationship between monoculture and applicant welfare becomes hazier. We note that our results do not address the concern that biases (such as racial or socioeconomic biases) can be exacerbated and made widespread by monoculture.

\paragraph{An Extension: Differential Application Access.} We extend our primary setup to include \textit{differential application access}---when different applicants submit a differing number of total applications (e.g., apply to a differing number of colleges). In practice, the number of applications submitted in college admissions varies significantly, and can be correlated with socioeconomic status: for example, in the 2021-2022 cycle, applicants who submitted 15 or more applications on the Common App were 2.5x as likely to have attended a private high school \citep{kim2022firstyear}. 

We are interested in the extent to which submitting additional applications improves the outcomes of an applicant. If submitting additional applications significantly improves outcomes for an applicant, this raises at least two concerns:
\begin{itemize}
    \item Disadvantaged applicants who submit fewer applications---for example, due to the high cost of application\footnote{In the 2022-2023 cycle, college applications in the United States cost up to \$100 \citep{wood2022colleges}.}---obtain worse outcomes than more advantaged counterparts with similar qualifications.
    \item Competitive effects force all students to submit as many applications as possible, even when this does not improve the welfare of applicants overall.
\end{itemize}
Our main result here (\Cref{thm:diff-app-access}) shows that applicants who submit more applications do not benefit under monoculture, but do benefit under polyculture. Intuitively, this is because an applicant's estimated value is the same across firms under monoculture, so submitting applications does not increase the likelihood that an applicant will have a higher estimated value. In fact, in the theoretical setting we study, it suffices for an applicant to apply to only their top choice: if their shared estimated value exceeds the shared cutoff, they will be matched to their most preferred option among those they applied to; otherwise, they will go unmatched regardless of where they applied. This is to be contrasted with polyculture, under which submitting more applications increases the likelihood that an applicant will receive a higher estimated value, since estimated values are drawn independently.

Our results also suggest that differential application access can lower the welfare of firms and colleges, since they are collectively less able to distinguish those who submitted many applications from those with higher values but who submitted fewer applications. An implication is that in order to obtain the benefits of polyculture (better distinguishing students with higher and lower value) it is important to ensure that applicants are on a level playing field---that different applicants submit a similar number of total applications.

Broadly speaking, our results suggest that monoculture is more robust than polyculture with respect to differential application access. This appears true from both sides of the market.

\subsection*{Overview of Computational Experiments}

Our theoretical results are developed in a setting with strong symmetry: applicants have uniformly random preferences over firms. In computational experiments, we test our main theoretical predictions across settings where applicant preferences have varying structure and degree of correlation. We find that the directional predictions of the theoretical results hold across these settings.
\section{Model and Preliminaries}
Our model builds on the model of \cite{azevedo2016supply}, in which there is a continuum of applicants and a finite number of firms (analogously, colleges). An applicant is identified by their true \textit{value}, which is distributed on $\RR$ according to a probability measure $\eta$. Let $\colleges = \{1, 2, \cdots, \numcolleges\}$ be the set of all firms. Firms have total capacity $\totalsupply < 1$ (so there are fewer positions than applicants) and each firm has the same capacity $\frac{\totalsupply}{\numcolleges}.$

An applicant with value $v$ is associated with an independent random variable $\theta(v)$, characterizing their preferences over firms and firms' estimated values for that applicant. The realization of $\theta(v)$ is their \textit{type}, which lies in $\Theta$. $\Theta := \mathcal{R} \times \RR^\numcolleges$ is the set of applicant types, where $\mathcal{R}$ is the set of preference orderings over the firms. If $\theta(v) = \theta = (\succ^\theta, e^\theta),$ then $\succ^\theta$ is the preference ordering of $v$ over firms and $e^\theta_\college$ is the \textit{estimated value} of the applicant at firm $\college$.

In the model, $\theta(v)$ plays a central role. First, it is used to model the process of preference approximation for firms. Firms all prefer applicants with higher true values, but in practice, rank applicants based on the estimated values given by $\theta(v)$. Second, it connects our model with that of \cite{azevedo2016supply}. In particular, the distribution $\eta$ over values in $\RR$ together with an appropriate choice of $\theta(v)$ induces a distribution over applicant types, which, alongside specifications of firm capacities, fully determines the model in \cite{azevedo2016supply}.\footnote{Formally, consider the probability distribution $\eta'$ over $\Theta$ such that for $A\subseteq \Theta$,
\begin{equation}
    \eta'(A) = \int_\RR \Pr[\theta(v)\in A]\,d\eta(v).
\end{equation}}

\paragraph{}
A \textit{matching} is a function $\mu: \Theta \rightarrow \colleges\cup \{\emptyset\}$ that satisfies the capacity constraints
\begin{equation}
    \int_\RR \Pr[\mu(\theta(v)) = \college]\,d\eta(v) = \frac{\totalsupply}{\numcolleges},\qquad \forall \college\in \colleges,
\end{equation}
as well as two additional technical assumptions: that $\mu^{-1}(\college)$ is measurable and that the set $\{\theta:\mu(\theta)\prec^\theta \college\}$ is open.\footnote{We remark that matchings are not typically required to fill capacities (i.e., the equality above is typically an inequality), but in our setting, since there is limited total capacity and applicants all prefer being matched to being unmatched, we are only concerned with matchings that fill capacity.}

Then $\mu(\theta(v))$ is a random variable indicating where an applicant with value $v$ is matched, with $\mu(\theta(v)) = \emptyset$ indicating that the applicant is not matched. When it is clear what $\theta$ is, we will abuse notation to write $\mu(v)$ instead of $\mu(\theta(v))$ to denote the random variable representing where an applicant with value $v$ is matched. Notice that the matching of an applicant depends only on their random type, reflecting the fact that their matching outcome depends directly on the scores they receive at each firm rather than their true value.

A matching is \textit{stable} if there does not exist a pair $\theta\in \Theta, \college\in \colleges$ such that $\college\succ^\theta \mu(\theta)$, and $e^\theta_\college > e^{\theta'}_{\college}$ for some $\theta'\in \mu^{-1}(\college)$. In other words, there is no applicant-firm pair that would jointly prefer to defect from the current matching to be matched with each other.

\subsection{A cutoff characterization of stable matchings}
A benefit of this continuum model is that stable matchings are characterized by a tractable cutoff structure \citep{azevedo2016supply}. Given a vector of \textit{cutoffs} $P = (P_1, \cdots, P_\numcolleges),$ an applicant with type $\theta$ can \textit{afford} firm $\college$ if $e^\theta_\college \ge P_\college$. The applicant's \textit{demand} at $P$ is $D^\theta(P)$, the applicant's most preferred firm (according to $\succ^\theta$) among those they can afford. If they cannot afford any firm, $D^\theta(P)=\emptyset$.
The \textit{aggregate demand} of a firm $\college$ at $P$ is
\begin{equation}
    D^\college(P) := \int_\RR \Pr[D^{\theta(v)}(P) = \college]\,d\eta(v),
\end{equation}
the measure of applicants who demand it.

A vector of cutoffs $P$ is \textit{market clearing} if and only if
\begin{equation}
    D^\college(P) = \frac{\totalsupply}{\numcolleges},\qquad \forall \college\in \numcolleges.\footnote{Again, as in the above footnotes, the standard definition in \cite{azevedo2016supply} would allow for an inequality here if no applicant would prefer to take the empty spots of a firm. This is not possible in our model, so we simplify our definition.}
\end{equation}

The following lemma then demonstrates that stable matchings are in one-to-one correspondence with market-clearing cutoffs.

\begin{lemma}[Supply and Demand Lemma \citep{azevedo2016supply}]\label{lem:supply-demand}
    A matching $\mu$ is stable if and only if there exists a vector of market-clearing cutoffs $P$ such that
    \begin{equation}
        \mu(\theta) = D^\theta(P).
    \end{equation}
\end{lemma}

The characterization of stable matchings given by the continuum model can be summarized as follows. Applicants each have an estimated value at each firm. Firms each have cutoffs. An applicant is then matched to their most preferred firm among those at which an applicant's estimated value is above the cutoff.

Next, we show how we can instantiate monoculture and polyculture within the model.

\subsection{Instantiating monoculture and polyculture}
We instantiate our model by specifying $\theta(v)$ for monoculture and polyculture. In what follows, we consider a fixed \textit{noise distribution} $\DD$.

\paragraph{Monoculture.} In monoculture, we take the random variable
\begin{equation}
    \theta_{\mono}(v) := (\succ, (v + X, v + X, \cdots, v + X)),
\end{equation}
where $\succ$ is drawn uniformly at random from $\mathcal{R}$ and $X\sim \DD.$ In other words, applicants have preferences distributed uniformly at random, and each applicant has a single estimated value $v+X$ that is shared across firms. This noisy estimated value equals the applicant's value perturbed by noise drawn from $\DD$. 

\paragraph{Polyculture.} In polyculture, we take the random variable
\begin{equation}
    \theta_{\poly}(v) := (\succ, (v + X_1, v + X_2, \cdots, v + X_\numcolleges)),
\end{equation}
where $\succ$ is drawn uniformly at random from $\mathcal{R}$ and $X_1,X_2,\cdots, X_\numcolleges \sim \DD$ are i.i.d.. Again, applicants have preferences distributed uniformly at random, but now each applicant has a distinct, independently drawn estimated value at each firm.

\paragraph{} For expository purposes, we instantiate our model with a fixed noise distribution $\DD$ shared across both monoculture and polyculture. However, as will often be obvious, this assumption is not necessary for many of our results. In some cases, (primarily, in Theorem 2) it does allow us to compare monoculture and polyculture \textit{ceteris paribus}.

\paragraph{A technical assumption on $\eta$ and $\DD$.} We will make a basic assumption on the probability measure $\eta$ and the probability distribution $\DD$. Let $\pi$ be the probability measure associated with $\DD$. Then we assume that both $\eta$ and $\pi$ have \textit{connected support}, where we mean that the smallest closed set of measure $1$ is an interval. Let these respective intervals be $[V_-, V_+]$ and $[X_-, X_+]$. The assumption of connected support means that there are no ``gaps'' in either distribution. This is theoretically convenient, since it ensures that there is a unique stable matching.

\subsection{Preliminary Analysis}

We now present some important preliminary analysis. The following important result implies that under both monoculture and polyculture, stable matchings take a simple form characterized by identical cutoffs at each firm.
\begin{lemma}[Equal Cutoffs Lemma]\label{prop:equal-cutoffs}
    If $\theta\in \{\theta_{\mono}, \theta_{\poly}\},$ then there is a unique vector of market-clearing cutoffs $P$, and
    $P_1 = P_2 = \cdots = P_\numcolleges.$
\end{lemma}

We denote these unique \textit{shared cutoffs} by $P_{\mono}$ and $P_{\poly}$, respectively, with corresponding stable matchings $\mu_{\mono}$ and $\mu_{\poly}$. Recall that under our abuse of notation, $\mu_{\mono}(v) := \mu_{\mono}(\theta_{\mono}(v))$ and $\mu_{\poly}(v) := \mu_{\poly}(\theta_{\poly}(v))$ are the random variables representing where $v$ is matched under monoculture and polyculture.

\paragraph{}
This cutoff structure allows for clean analysis of stable matchings. Before moving to our main results, we make some basic observations that follow. The next proposition follows directly from the Supply and Demand Lemma (\Cref{lem:supply-demand}) and the Equal Cutoffs Lemma (\Cref{prop:equal-cutoffs}).

\begin{proposition}\label{prop:prob-match}
    For all $v \in \RR$,
    \begin{enumerate}
    \item[(i)] the probability an applicant with value $v$ is matched under monoculture is
\begin{equation}\label{eq:match-mono}
    \Pr[\mu_{\mono}(v) \in \colleges] = \Pr[v + X > P_{\mono}],
\end{equation}
where $X\sim \DD,$ and
\item[(ii)] the probability an applicant with value $v$ is matched under polyculture is
\begin{equation}\label{eq:match-poly}
        \Pr[\mu_{\poly}(v) \in \colleges] = \Pr[v + \max_{\college\in \colleges} X_\college > P_{\poly}],
    \end{equation}
    where $X_1,\cdots, X_\numcolleges\simiid \DD.$
\end{enumerate}
\end{proposition}

The above two propositions contrast the structure of stable matchings in the monoculture and polyculture economies. In monoculture, whether an applicant matches to a firm depends on a single noisy estimated value. In polyculture, whether an applicant matches depends on the \textit{maximum} of $\numcolleges$ noisy estimated values.

We can also use \Cref{prop:prob-match} to show that the shared cutoff is lower under monoculture than polyculture, which will be useful in later analysis. This explains the puzzle of why it is not the case that fewer applicants are matched overall under monoculture: when estimated values at firms become more correlated, the cutoff required for acceptance decreases.
\begin{corollary}[Lower Cutoff Under Monoculture]\label{cor:cutoff-inequality}
When $\numcolleges \ge 2,$ the shared cutoff in an economy is lower under monoculture than under polyculture:
    \begin{equation}
        P_{\mono} < P_{\poly}.
    \end{equation}
\end{corollary}
Intuitively, the result holds because under polyculture, applicants have multiple chances at getting an estimated value that is higher than the shared cutoff; so in order for the same number of applicants to be matched under monoculture and polyculture, the cutoff under polyculture must be lower.

\subsection{Maximum-Concentrating Noise}

We now introduce two key definitions that will be useful for reasoning about the noise distribution $\DD$. In particular, we define a particular class of noise distributions that we will focus our later analysis on.

\begin{definition}[Maximum Order Statistic]
    The \textbf{maximum order statistic} of $n$ draws from a distribution $\DD$ is defined as the random variable
    \begin{equation}
        X^{(n)} := \max \{X_1, X_2, \cdots, X_n\}
    \end{equation}
    where $X_1, \cdots, X_n\simiid \DD.$
\end{definition}

The maximum order statistic plays an important role in our analysis, since the probability an applicant with value $v$ is matched under polyculture \eqref{eq:match-poly} can be rewritten as
\begin{equation}\label{eq:match-poly-order}
    \Pr[\mu_{\poly}(v) \in \colleges] = \Pr[v + X^{(\numcolleges)} > P_{\poly}].
\end{equation}

\begin{definition}[Maximum-Concentrating Distribution]\label{def:maximum-concentrating}
    A distribution $\DD$ is \textbf{maximum concentrating} if for all $\epsilon > 0$,
    \begin{equation}
        \lim_{n\rightarrow \infty} \Pr\left[|X^{(n)} - \EE[X^{(n)}]| > \epsilon\right] = 0.
    \end{equation}
\end{definition}
In other words, a distribution is maximum concentrating if its maximum order statistic satisfies a ``weak law of large numbers.''

\begin{example}[The Uniform Distribution is Maximum Concentrating]\label{ex:uniform}
    Intuitively, the uniform distribution $U[0,1]$ is maximum concentrating because the maximum of $n$ draws from $U[0,1]$ is almost certainly near $1$ when $n$ grows large.\footnote{We can make this intuition precise. The maximum order statistic of $n$ draws from $U[0,1]$ is distributed according to the beta distribution $\Beta(n,1)$, which has mean $1 - \frac{n}{n+1}$ and variance $\frac{n}{(n+1)^2(n+2)}.$ Therefore,
    \begin{align}
        \EE[X^{(n)}] &= \frac{n-1}{n}\\
        \Var[X^{(n)}] &= \frac{n}{(n+1)^2(n+2)}.
    \end{align}
    To show that $U[0,1]$ is maximum concentrating, it suffices to apply Chebyshev's inequality:
    \begin{equation}
        \Pr\left[\left|X^{(n)} - \EE[X^{(n)}]\right| > \epsilon \right] < \frac{\Var[X^{(n)}]}{\epsilon} = \frac{n}{\epsilon(n+1)^2(n+2)}.
    \end{equation}
    The result follows since $\lim_{n\rightarrow \infty} \frac{n}{\epsilon(n+1)^2(n+2)} = 0.$}
\end{example}

It is possible to show that all bounded distributions are maximum-concentrating. Being bounded is not, however, a necessary condition---the Gaussian distribution is also maximum-concentrating despite $\EE[X^{(n)}]$ growing arbitrarily large as $n$ increases. Roughly speaking, maximum-concentrating distributions are those with an upper tail that is ``lighter than exponential'' (i.e., has a tail that is lighter than the exponential or Laplace distribution).

In general, it is possible to determine if a distribution is maximum concentrating using the Fisher-Tippett-Gnedenko theorem, a cornerstone of Extreme Value Theory.\footnote{See \cite{haan2006extreme} for a detailed exposition.} As in \Cref{ex:uniform}, it also suffices to show that the variance of a distribution's maximum order statistic converges to $0$ as $n\rightarrow \infty$; the result then follows directly from Chebyshev's inequality.
\section{Main Results: Firm and Applicant Welfare}\label{sec:results}

\subsection{Firm welfare}\label{sec:results-college}
In this section, we focus on the perspective of firms, who would like to admit the most-preferred applicants (i.e., those with the highest true values). To do this, we analyze---under monoculture and polyculture---the distribution of applicants who are matched. Specifically, we determine the probability that an applicant with a given value is matched (this, of course, is also of interest to the applicant). This allows us to find the expected value of matched applicants.

As it turns out, strong behavior emerges under polyculture as the number of firms $\numcolleges$ in the economy grows large. In particular, when the noise distribution $\DD$ is maximum-concentrating (see \Cref{def:maximum-concentrating}), only the highest-value applicants are matched; as a function of an applicant's value, the probability that an applicant is matched approaches a step function when $\numcolleges$ approaches $\infty$. This, in turn, implies that the matching under polyculture approaches optimal.

Define $v_S$ to be the unique value such that $\eta((v_S, \infty)) = S$, i.e., the threshold at which the mass of students with a value above $v_S$ equals the overall firm capacity. Then note that in the firm-optimal matching, only applicants with value greater than $v_S$ are matched. We may now state our first main result. 

\begin{theorem}[Wisdom of the Crowds in Polyculture but not Monoculture]\label{thm:wisdom}
    Suppose $\DD$ is maximum concentrating. Then:
    \begin{enumerate}
        \item[(i)] Under polyculture, the probability an applicant is matched approaches $0$ or $1$ as $\numcolleges\rightarrow \infty$:
            \begin{equation}
                \lim_{m\rightarrow \infty} \Pr[\mu_{\poly}(v)\in \colleges]
                = \begin{cases}
                    0 &\quad v < v_S \\
                    1 &\quad v > v_S
                \end{cases}
                \quad,
            \end{equation}
            and firm welfare approaches optimal.
        \item[(ii)] Under monoculture, the probability an applicant is matched
        \begin{equation}
            \Pr[\mu_{\mono}(v)\in \colleges]
        \end{equation}
        is constant in $\numcolleges$, and firm welfare is constant and suboptimal.
    \end{enumerate}
\end{theorem}

\begin{figure}
    \centering
    \subfloat{{\includegraphics[width=0.49\linewidth]{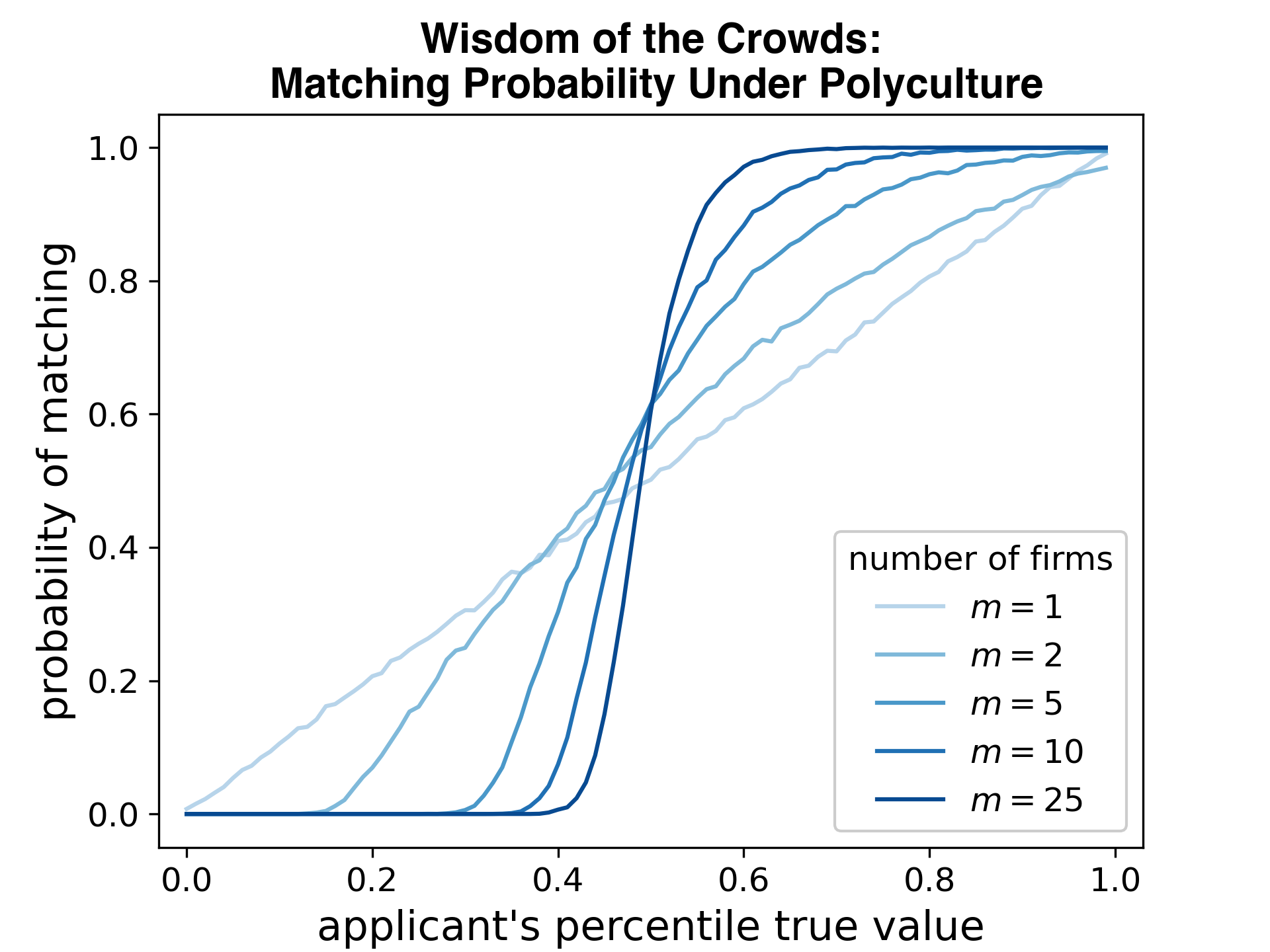}}}
    \subfloat{{\includegraphics[width=0.49\linewidth]{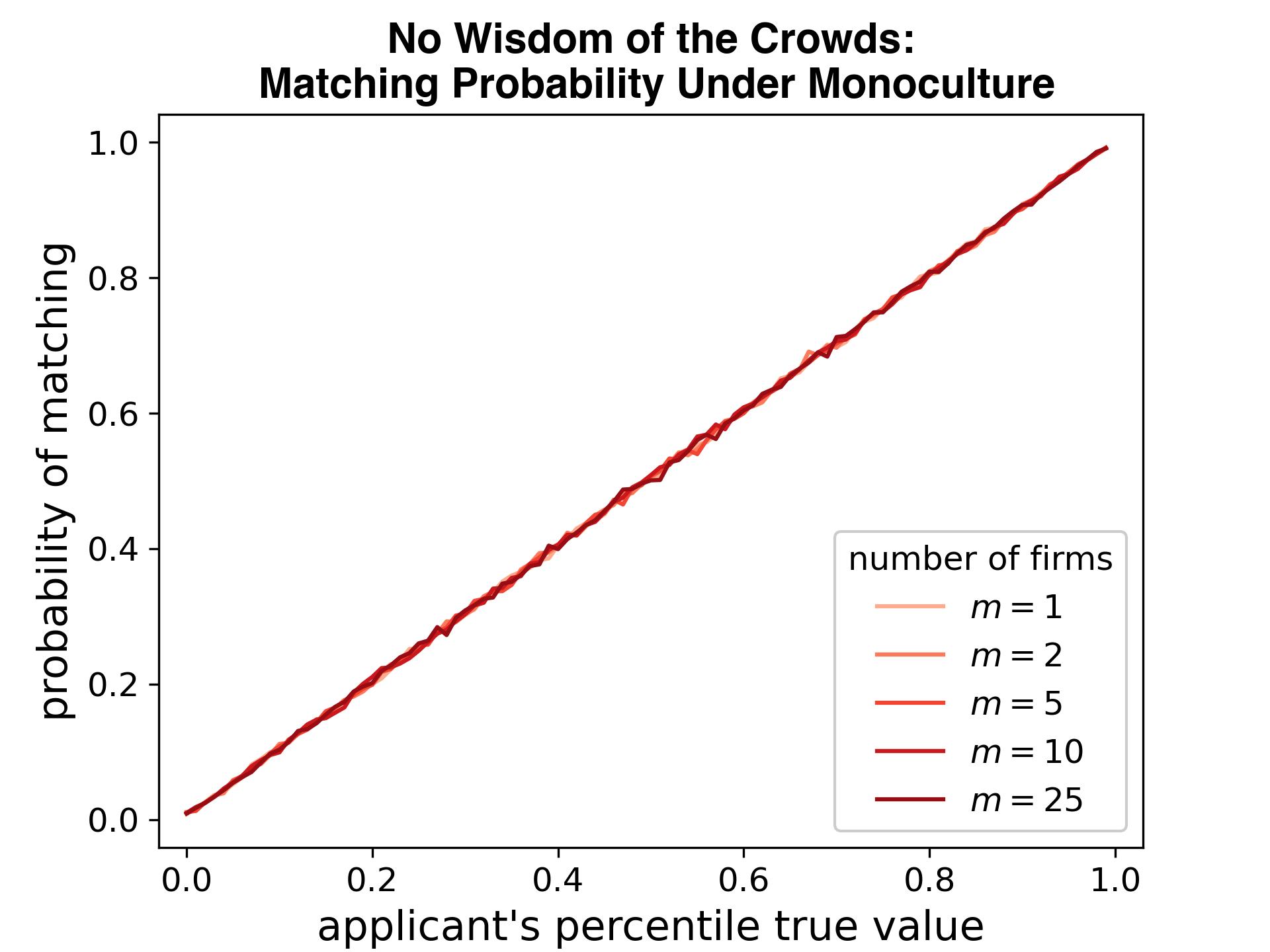}}}
    \caption{As the number of firms $\numcolleges$ increases, $\Pr[\mu_{\poly}(v)\in \colleges]$ approaches a step function, meaning that only the highest-value applicants are matched. Meanwhile, $\Pr[\mu_{\mono}(v)\in \colleges]$ remains the same regardless of the number of firms; in particular, high-value applicants may have positive probability of being unmatched. There are 1000 total applicants and firms have total capacity 500. Applicants have values drawn uniformly from $[0,1]$ and uniformly-random preferences. Noise is drawn uniformly from $[-0.5, 0.5].$ The plot is generated by taking the average over 1000 simulations.}
    \label{fig:wisdom}
\end{figure}

Despite individual firms only having noisy preferences, under polyculture, they collectively admit only the highest-value applicants as $\numcolleges$ grows large (as illustrated in \Cref{fig:wisdom}). This gives a ``wisdom of the crowds''-like result, where individually noisy information is pooled together to provide a more accurate assessment. The remarkable feature of our context is that firms do not directly share information with one another; rather, it is the market at large (specifically, the structure of stable matchings) that facilitates the pooling of information.\footnote{This holds even despite an apparent ``winner's curse'' \citep{capen1971competitive, thaler1988anomalies, cox1984search}: firms who extend an offer to an applicant tend to have an upward biased estimate of the applicant's value.}

Taking a step back, \Cref{thm:wisdom} offers a comparison between monoculture and polyculture. For any maximum-concentrating distribution $\DD=\DD_{\poly}$, as the number of firms grows large, only the highest-value applicants are matched under polyculture. On the other hand, for \textit{any} distribution $\DD=\DD_{\mono}$, there will always be some ``error'' in who is matched under monoculture, and the amount of error does not depend on the number of firms. Thus, from \Cref{thm:wisdom}, we can predict that polyculture---when there are sufficiently many firms---will produce a matching with higher firm welfare than monoculture. Notably, this observation does not depend on the relationship between the accuracy of the noise distributions in monoculture and polyculture; even when a shared algorithm is significantly more accurate than firms' independent evaluations, firms can be better off all using these independent evaluations.

\begin{figure}
    \begin{center}
        \includegraphics[width=0.9\linewidth]{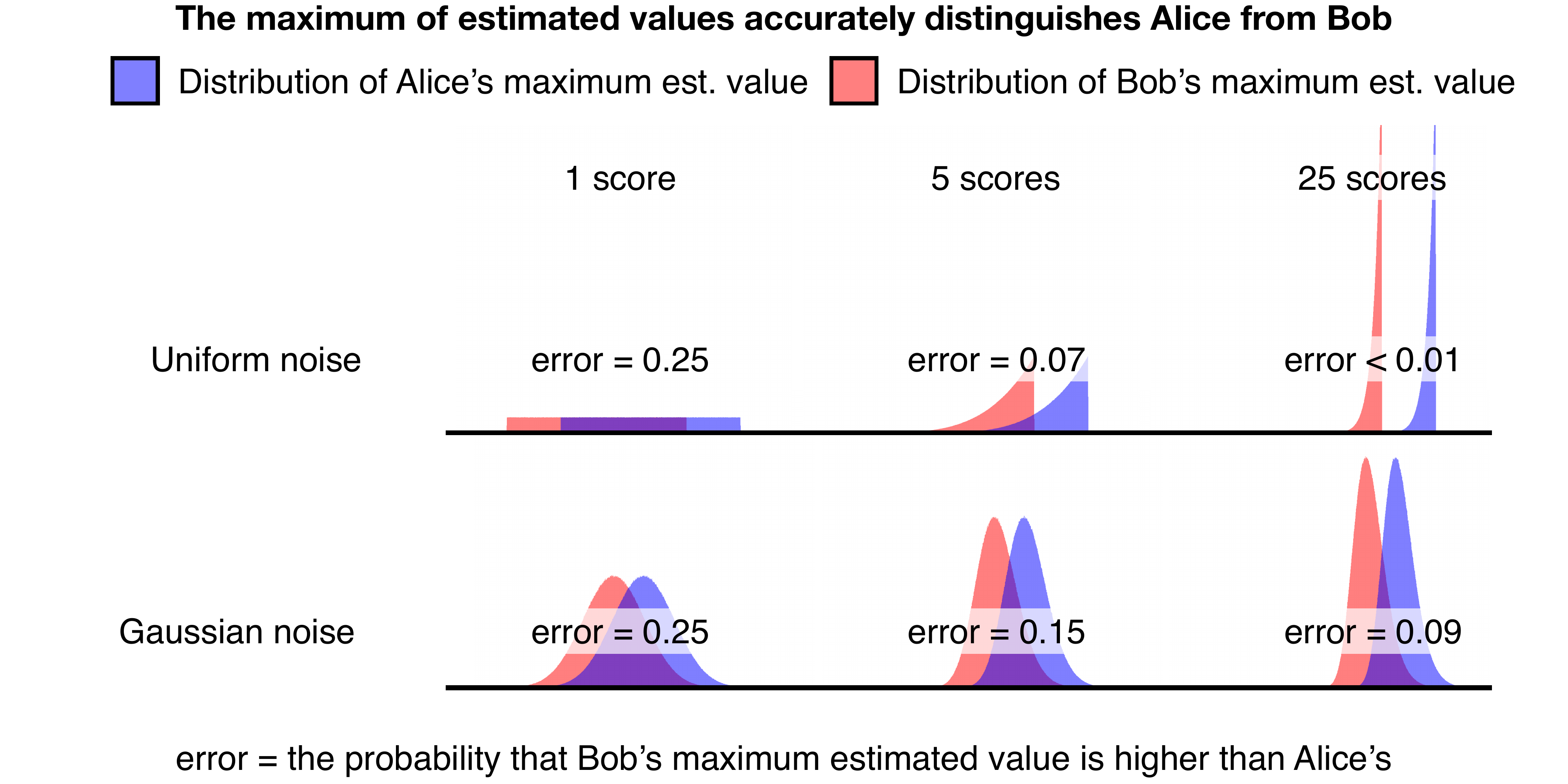}
    \end{center}
    \caption{Suppose Alice's true value is higher than Bob's. Estimated values are noisy, such that given a single estimated value, Bob's estimated value has a 25\% chance of being higher than Alice's. However, when considering several independently drawn estimated values, the probability that Bob's \textit{maximum} estimated value exceeds Alice's vanishes in the number of estimated values drawn. As shown, this effect is more gradual for the Gaussian distribution than the uniform distribution, though both are maximum-concentrating.}
    \label{fig:max-concentrating}
\end{figure}

\paragraph{Intuition for \Cref{thm:wisdom}.}
We now give some intuition for \Cref{thm:wisdom}(i), which at a high level, is that the maximum estimated value---which is what determines if an applicant is matched---is able to effectively ``distinguish'' between applicants of higher and lower values. Consider applicants Alice and Bob. Alice's value $v_A$ is higher than Bob's value $v_B$. We show that if $\DD$ is maximum-concentrating, the probability that Bob is matched but Alice is unmatched vanishes as the number of firms $\numcolleges$ grows large. If Bob is matched but Alice is not, then Bob's maximum estimated value must be higher than Alice's. The probability this occurs is given by
\begin{equation}
    \Pr[v_A + X_A^{(\numcolleges)} < v_B + X_B^{(\numcolleges)}],
\end{equation}
where $X_A^{(\numcolleges)}$ and $X_B^{(\numcolleges)}$ are independent variables distributed according to $X^{(m)}$. This reduces to
\begin{equation}
    \Pr[X_B^{(\numcolleges)} - X_A^{(\numcolleges)} > v_A - v_B].
\end{equation}
Since $\DD$ is maximum concentrating, the probability that two independent draws from $X^{(\numcolleges)}$ differ by more than the constant $v_A - v_B$ vanishes as $\numcolleges$ grows large. Thus, Alice's maximum estimated value will almost always be higher than Bob's. This is illustrated in \Cref{fig:max-concentrating}, which shows the distribution of Alice and Bob's maximum estimated values depending on the number of firms $m$.

\subsection{Applicant welfare}\label{sec:results-student}
We now consider matching from the applicant perspective in greater detail. Our main result (\Cref{thm:top-choice}) shows a few things. First, it shows that all applicants are more likely to be matched to their top choice firm under monoculture than under polyculture (part (i))---and moreover, that, under monoculture, all applicants who are matched to \textit{some} firm are matched to their \textit{top} firm (part (ii)). This implies that total applicant welfare---as measured by the average rank of where applicants are matched---is optimal under monoculture (and better than under polyculture). This is illustrated in \Cref{fig:match-rank}.

These findings have a counterintuitive flavor to it. In particular, it challenges the prevalent intuition that from the perspective of an applicant, monoculture is worrisome due to the increased possibility of being systematically shut out: being rejected from one means being rejected from all. But because the total number of applicants matched remains the same, the proportion of applicants facing this systematic rejection is the same in any case. On the flip side, it is true that an offer from one firm implies an offer from all firms under monoculture, which is not true under polyculture. This is true even though the number of applicants receiving \textit{at least one} offer remains the same. This property implies that under monoculture, more applicants receive their top choice.

Meanwhile, in part (iii), we show that for some applicants, even if they are just as likely or more likely to match to their top choice under monoculture, they are less likely to be matched overall. In fact, there exists a set of applicants of positive measure that is simultaneously more likely to match to their top choice under monoculture while less likely to match overall. For these applicants, neither monoculture or polyculture stochastically dominates the other.

\paragraph{}
Proceeding to the statement of the result, we define $\Rank_{\theta(v)}(\college) := |\{\college':\college'\succeq^{\theta(v)} \college\}|$ to be the rank of an applicant's matched firm. (We let $\Rank_{\theta(v)}(\emptyset) = 0$.) Then $\Rank_{v}(\mu(v))$ is a random variable giving the rank of the (random) firm an applicant of value $v$ is matched to under $\mu$. Also recall that $(X_-, X_+)$ is the support of the noise distribution $\DD$.

\begin{figure}
    \centering
    \includegraphics[width=10cm]{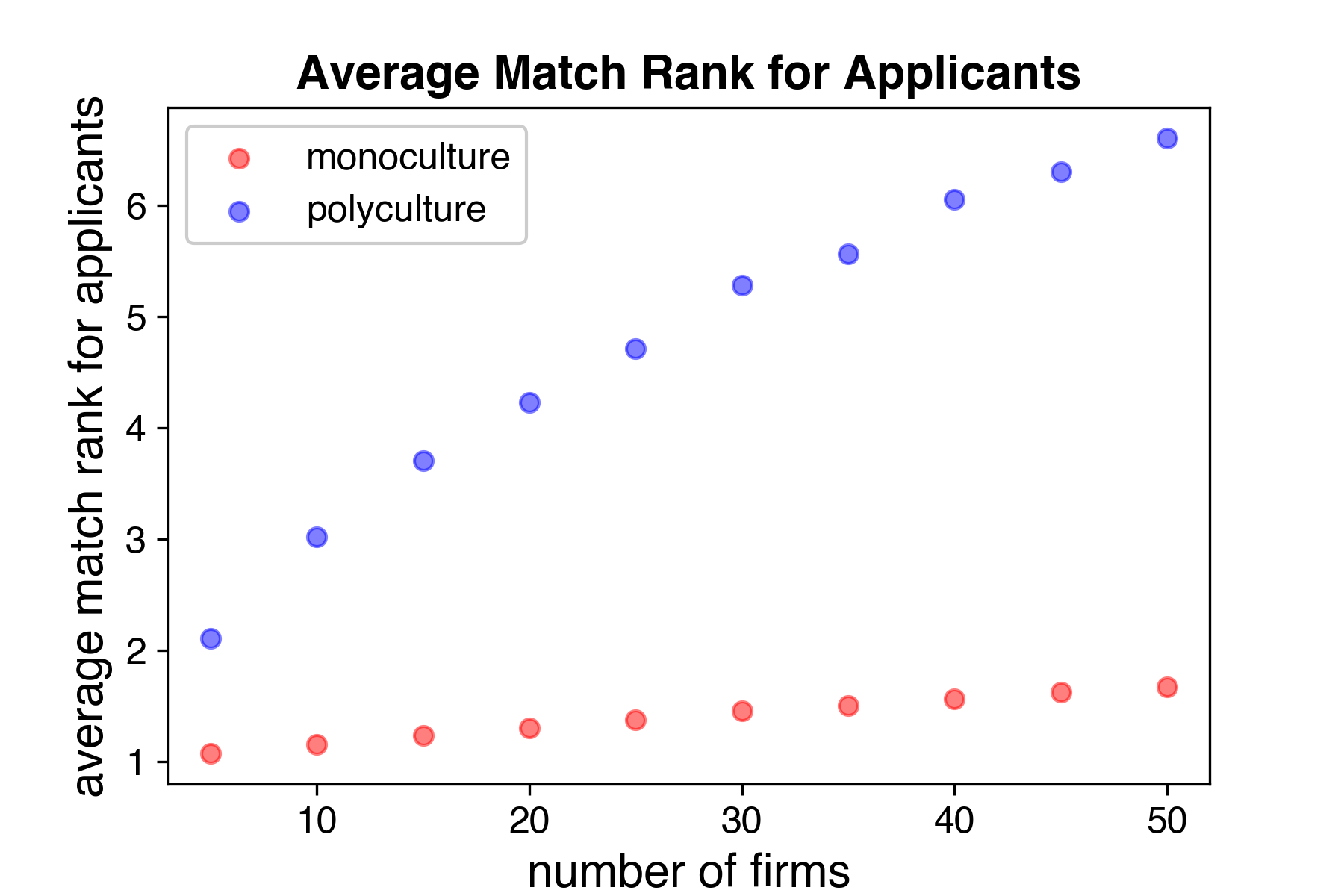}
    \caption{Applicants receive more-preferred options under monoculture. There are 1000 total applicants and firms have total capacity 500. Thus, if there are $\numcolleges$ firms, each has capacity $\frac{500}{\numcolleges}$. Applicants have values drawn uniformly from $[0,1]$ and uniformly-random preferences. Noise is drawn uniformly from $[-0.5, 0.5].$ In simulation, the slight increase in the average rank of matches under monoculture is due to stochasticity in the finite setting; cutoffs are not exactly the same across firms, so a small number of applicants matched to one firm are not matched to their top firm. The plot is generated by taking the average over 100 simulations.}
    \label{fig:match-rank}
\end{figure}

\begin{theorem}[Likelihood of Matching to Top Choice or At All]\label{thm:top-choice}
The following hold:
\begin{enumerate}
    \item[(i)] For all $v$, the probability an applicant matches to their top choice is at least as high under monoculture as under polyculture:
    \begin{equation}\label{eq:prob-first}
        \Pr[\Rank_v(\mu_{\mono}(v)) = 1] \ge \Pr[\Rank_v(\mu_{\poly}(v)) = 1],
    \end{equation}
    and the inequality is strict for all $v\in (P_{\mono} - X_+, P_{\mono} - X_-),$ a set of positive $\eta$-measure.
    \item[(ii)] For all $v$, if $\mu_{\mono}(v)\in \colleges$,
    \begin{equation}
        \Rank_v(\mu_{\mono}(v)) = 1,
    \end{equation}
    meaning that all applicants who match under monoculture are matched to their top choice, and that the matching attains optimal applicant welfare.
    \item[(iii)] 
    When $\DD$ is maximum concentrating, for all $v\in (v_S, P_{\mono} - X_-),$ which is a set of positive $\eta$-measure, and for all $\numcolleges$ sufficiently large, an applicant with value $v$ is less likely to match under monoculture than under polyculture:
    \begin{equation}\label{eq:tiger-2}
        \Pr[\mu_{\mono}(v)\in \colleges] < \Pr[\mu_{\poly}(v)\in \colleges].
    \end{equation}
\end{enumerate}
\end{theorem}

The high-level reasoning behind these results is fairly straightforward. First, in part (i), the probability an applicant is matched to their top choice is exactly the probability that their estimated value at that firm exceeds the shared cutoff: $P_{\mono}$ under monoculture and $P_{\poly}$ under polyculture. We know that $P_{\mono}<P_{\poly}$ (from \Cref{cor:cutoff-inequality}), so this probability is at least as large under monoculture in comparison to under polyculture (and in many cases, strictly larger).

To show part (ii), observe that if an applicant's common estimated value under monoculture exceeds the shared cutoff $P_{\mono}$, that means that the applicant can afford any firm; thus, they are matched to their top choice.

Finally, part (iii) is a consequence of \Cref{thm:wisdom}, where we showed that for applicants with value $v$ above a certain threshold $v_S$, their likelihood of being matched approaches $1$ as $m\rightarrow \infty$, while their likelihood of being matched remains constant (and less than $1$) under monoculture when $v < P_{\mono} - X_-$. (Applicants with true value above $P_{\mono} - X_-$ are guaranteed to be matched under monoculture, and thus do not gain an advantage from polyculture in this respect. Even if they receive the worst possible noise $X_-$, their estimated value still exceeds the threshold $P_{\mono}$.)

We remark that part (i) depend on the noise distribution $\DD$ being the same under monoculture and polyculture. For example, in part (i), consider the extreme case in which there is no noise under polyculture, but substantial noise under monoculture. Then all of the top applicants are guaranteed to be matched to their top choice under polyculture, and may not be matched at all under monoculture. It follows that, for these top applicants, the probability of matching to their top choice is higher under polyculture than monoculture in this setting. Part (iii) holds for distinct distributions $\DD_{\mono}$ and $\DD_{\poly}$ as long as $D_{\poly}$ is maximum-concentrating and $(X_-, X_+)$ is defined to be the support of $\DD_{\mono}.$

\section{An Extension: Differential Application Access}\label{sec:diff-access}

We now consider when different applicants can apply to different numbers of firms (analogously, colleges)---which we refer to as \textit{differential application access}. In particular, we assume some discrete probability distribution $\kappa$ over $\{1, 2, \cdots, \numcolleges\}$ that determines the number of firms an applicant can apply to. In particular, an applicant with value $v$ can apply to $k(v)$ firms, where $k(v)$ is drawn independently from $\kappa$. All other aspects of our model remain the same.

Under both monoculture and polyculture there is a natural Nash equilibrium in which applicants all apply to their top $k$ firms. In this equilibrium, we find that applicants who can apply to more firms benefit under polyculture, producing inequitable outcomes. This is not true under monoculture. Consequently, while firm welfare is not harmed when applicants can apply to a different number of firms under monoculture, it is harmed under polyculture. Together, these results suggest that monoculture is more robust than polyculture with respect to differential application access.

\subsection{Model setup}
We begin by instantiating differential application access within our model.
An applicant $v$ who can apply to $k$ firms must choose an application strategy (i.e., which $k$ firms to apply to). A strategy profile is given by a vector of maps $S = (S_1, S_2, \cdots, S_\numcolleges),$ where $S_k: \RR\times \mathcal{R}\rightarrow \{S\in 2^{\colleges}: |S|=k\}$. Therefore, an applicant who who has value $v$ and a preference ordering $\succ$ and who can apply to $k$ firms applies to the firms in $S_k(v, \succ).$ Therefore, a strategy profile uniquely determines where an applicant with a given value, preference list, and level of application access applies.

We will focus on the strategy profile $S$ such that $S_k(v, \succ)$ gives the $k$ most preferred firms according to $\succ$. Therefore, strategies do not depend at all on $v$. We will later show that this intuitive set of strategies forms a Nash equilibrium. Fixing this strategy profile $S$, we can now define $\theta(v)$ for monoculture and polyculture in the presence of differential application access.

\paragraph{Monoculture with differential application access.} Define
\begin{equation}
    \theta_{\mono, \kappa}(v) = (\succ, e)
\end{equation}
where
\begin{equation}
    e_c = 
    \begin{cases}
        v + X &\quad \college\in S_{k(v)}(v, \succ)\\
        -\infty &\quad \college\notin S_{k(v)}(v, \succ)
    \end{cases},
\end{equation}
for $\succ$ drawn uniformly at random from $\mathcal{R}$, $k(v)$ drawn from $\kappa,$ and $X$ drawn from $\DD$. As before, an applicant with value $v$ is given a random preference ordering over firms. However, now the applicant's estimated values across firms depends also on a random draw from $\kappa$, which determines how many firms they can apply to. Recall that $e$ gives the vector of estimated values of an applicant across firms. Therefore, under monoculture, this estimated value is shared across the firms the applicant applies to (which are the applicant's top-$k(v)$ most preferred firms), but is $-\infty$ at other firms (meaning that the applicant will never be matched to firms they do not apply to).

\paragraph{Polyculture with differential application access.} Define
\begin{equation}
    \theta_{\poly, \kappa}(v) = (\succ, e)
\end{equation}
where
\begin{equation}
    e_\college = 
    \begin{cases}
        v + X_\college &\quad \college\in S_{k(v)}(v, \succ)\\
        -\infty &\quad \college\notin S_{k(v)}(v, \succ)
    \end{cases},
\end{equation}
for $\succ$ drawn uniformly at random from $\mathcal{R}$, $k(v)$ drawn from $\kappa,$ and $X_1,\cdots,X_\numcolleges\simiid \DD.$ Here, an applicant gets an independently drawn estimated value at each firm they apply to.

\paragraph{} We now give the Equal Cutoffs Lemma in this setting, which again, is a consequence of symmetry in applicants' preferences over firms.
\begin{lemma}[Equal Cutoffs Lemma for Differential Application Access]\label{lem:equal-cutoffs-diff-access}
    If $\theta\in \{\theta_{\mono, \kappa}, \theta_{\poly, \kappa}\},$ then there is a unique vector of market-clearing cutoffs $P$, and
    $P_1 = P_2 = \cdots = P_\numcolleges.$
\end{lemma}
We denote these unique stable matchings by $\mu_{\mono, \kappa}$ and $\mu_{\poly, \kappa},$ with corresponding \textit{shared cutoffs} $P_{\mono, \kappa}$ and $P_{\poly, \kappa}$. Recall that under our abuse of notation, $\mu_{\mono, \kappa}(v) := \mu_{\mono, \kappa}(\theta_{\mono, \kappa}(v))$ and $\mu_{\poly, \kappa}(v) := \mu_{\poly, \kappa}(\theta_{\poly, \kappa}(v))$ are the random variables representing where $v$ is matched under monoculture and polyculture.

The strategy profile we consider here forms an ex-ante Nash equilibrium, meaning that no applicants---prior to knowing their estimated values at each firm---would apply to a set of firms different from their top $k$.
\begin{proposition}[Nash Equilibrium]\label{prop:nash}
    For any $\kappa$, and under both monoculture and polyculture, no applicant benefits ex-ante by deviating from $S$.
\end{proposition}
This is an immediate consequence of \Cref{lem:equal-cutoffs-diff-access}. Since cutoffs at all firms are the same, an applicant is strictly harmed by applying to a set of firms other than their top $k$, since doing so does not increase their chances of being matched to any given firm (or set of firms).

\subsection{Analysis}

\begin{figure}
    \centering
    \includegraphics[width=10cm]{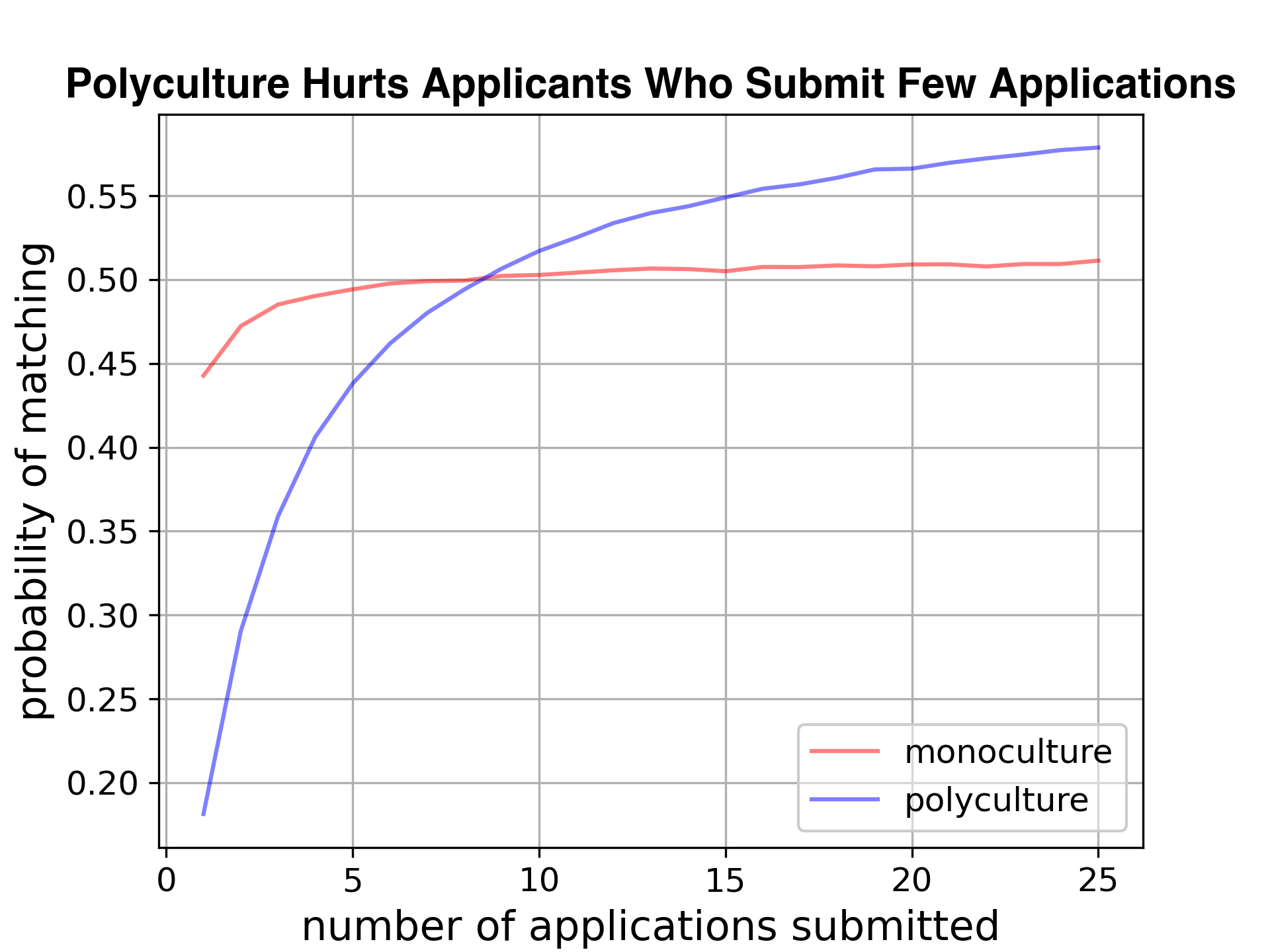}
    \caption{We plot the probability of matching conditional on $k$, the number of applications submitted. There are 1000 applicants and 25 firms have total capacity 500. Applicants have values drawn uniformly from $[0,1]$ and uniformly-random preferences. Noise is drawn uniformly from $[-0.5, 0.5].$ Each applicant can apply to $k$ colleges, drawn uniformly from $\{1,2,\cdots,25\}.$ Probabilities are averages over 10,000 simulations.}
    \label{fig:diff_access_num_apps}
\end{figure}

We now show that while monoculture is robust to differential application access, polyculture is not. Our main result in this section is that under monoculture, applicants who apply to more firms do not gain an advantage, while under polyculture, applicants who apply to more firms are more likely to be matched. Again recall that $(X_-, X_+)$ is the interval on which $\DD$ is supported.

\begin{theorem}[Differential Application Access in Monoculture and Polyculture]\label{thm:diff-app-access}
The following hold:
    \begin{itemize}
        \item[(i)] In monoculture, applicants who apply to more firms do not gain an advantage: For all $v$, 
        \begin{equation}
            \Pr[\mu_{\mono, \kappa}(v)\in \colleges\,|\,k(v)=k] = \Pr[\mu_{\mono}(v)\in \colleges]
        \end{equation} 
        is constant in $k$.
        \item[(ii)] In polyculture, applicants who apply to more firms gain an advantage: For all $v$, 
        \begin{equation}
           \Pr[\mu_{\poly, \kappa}(v)\in \colleges\,|\,k(v)=k] 
        \end{equation}
        is increasing in $k$, and strictly increasing in $k$ for $v\in (P_{\poly, \kappa} - X_+, P_{\poly, \kappa} - X_-).$
    \end{itemize}
\end{theorem}

This result is illustrated in \Cref{fig:diff_access_num_apps}, which shows that while the probability of matching does not depend significantly on $k$ under monoculture, it is increasing in $k$ under polyculture. A consequence of the result is that students with higher values can go unmatched while students with lower values but who submit more applications may be matched. From the firm (or college) perspective, this means that the ability to distinguish higher- and lower-value applicants is diminished. We demonstrate this in our computational experiments (see, e.g., \Cref{fig:exp-diff-access-college-welfare}). This further suggests that the strong result of \Cref{thm:wisdom} is not robust to differential application access, and as a consequence, firms and colleges under polyculture are incentivized to ``level the playing field'' to ensure that applicants each submit a similar number of total applications.

Notice that, as in \Cref{thm:wisdom}, the results here hold independently of the relationship between the noise distributions in monoculture and polyculture.

\section{Computational Experiments: Correlated Preferences}
In this section, we perform computational experiments to test the accuracy of predictions from our theoretical results in more general settings. We focus on our three broad theoretical findings: that all else equal, in comparison to polyculture, monoculture:
\begin{enumerate}
    \item[(1)] selects less-preferred applicants (\Cref{thm:wisdom}),
    \item[(2)] matches more applicants to their top choice, though effects for individual applicants vary depending on their value (\Cref{thm:top-choice}),
    \item[(3)] is more robust to differential application access (\Cref{thm:diff-app-access}).
\end{enumerate}

\paragraph{Correlated applicant preferences.}
We proved our theoretical results in a setting where applicants have uniformly random preferences over firms. We consider computational experiments in which applicants can have correlated preferences over firms. To introduce this correlation, we follow \cite{ashlagi2017unbalanced} in using a random utility model adopted from \cite{hitsch2010matching}.\footnote{In \cite{ashlagi2017unbalanced}, the model is used to generate preferences on both sides of the market; in our setting, we only use the model to generate applicant-side preferences, i.e., how applicants rank firms.} 

The model generates applicant preferences in the following way. Each applicant $i$ has a characteristic $x_i^D$ drawn independently from $U[0,1]$. Each firm has two characteristics, $x_\college^A$ and $x_\college^D,$ both drawn independently from $U[0,1].$ Then the utility applicant $i$ for being matched to firm $\college$ is
\begin{equation}
    u_i(\college) = \beta x_\college^A - \gamma(x_i^D - x_\college^D)^2 + \epsilon_{i\college},
\end{equation}
where $\epsilon_{i\college}$ is drawn independently from the standard logistic distribution. Here, $x_\college^A$ is a vertical measure of firm quality, shared by all applicants, so $\beta$ controls the level of correlation between preferences. $x_i^D$ and $x_\college^D$ are ``locations'' of the applicant and firm, and $\gamma$ controls the preference for an applicant to be ``close'' to the firm. $\epsilon_{i\college}$ accounts for other idiosyncratic factors.

When $\beta=\gamma=0,$ we recover the uniformly random preferences used in our theoretical analysis. When $\beta$ grows large, preferences become fully correlated: applicants all have the same preferences. When $\gamma$ increases, there is more correlation in preferences between ``nearby'' applicants.

\paragraph{Experiment details.}
We now consider a market in which there are 1000 applicants and 10 firms each with capacity 50. (So half of applicants are matched.) We let applicant values be distributed uniformly from $0$ to $1$ (so $\eta$ is the uniform measure on $[0,1]$), and take the noise distribution $\DD$ to be the Gaussian distribution $\mathcal{N}(0,\frac{1}{2})$. We vary $\beta$ and $\gamma$ between $0$ and $20$. The plots we show are all in this setting, and provide averages over 100 random instantiations of the market.

\begin{figure}
    \centering
    \subfloat{{\includegraphics[width=0.49\linewidth]{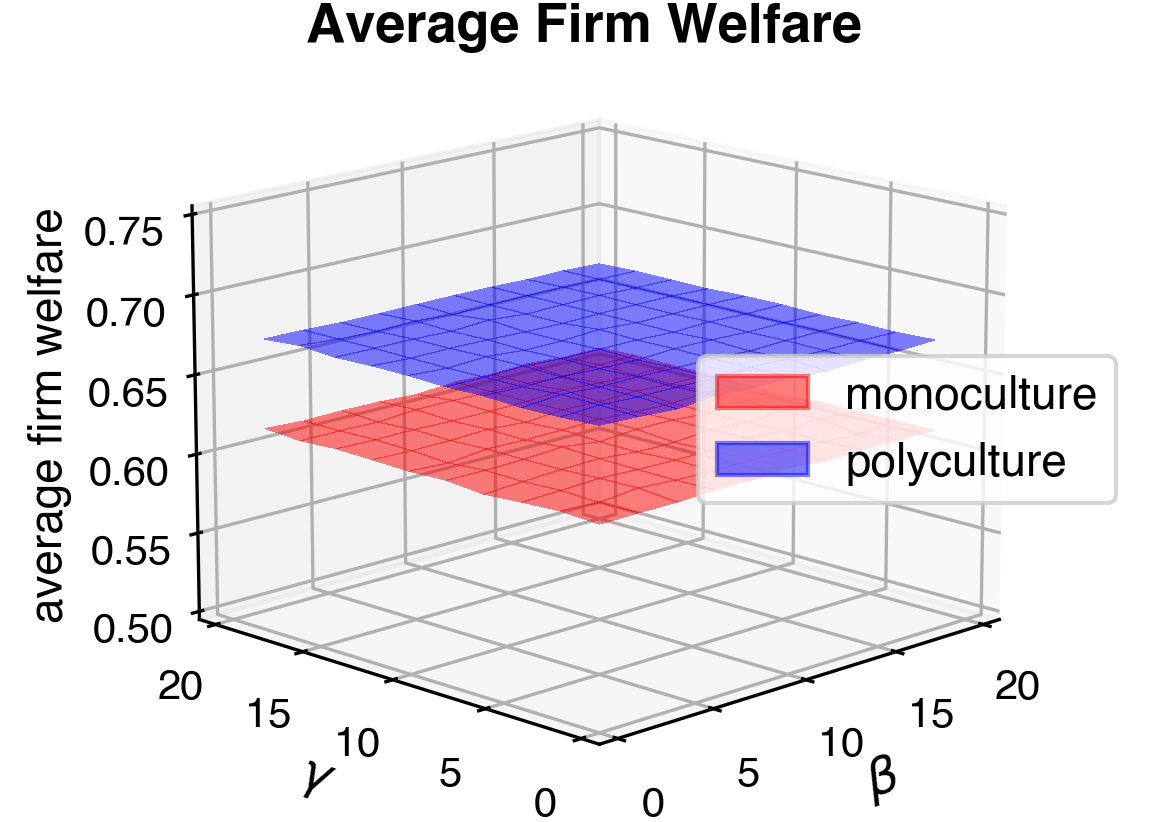}}}
    \subfloat{{\includegraphics[width=0.49\linewidth]{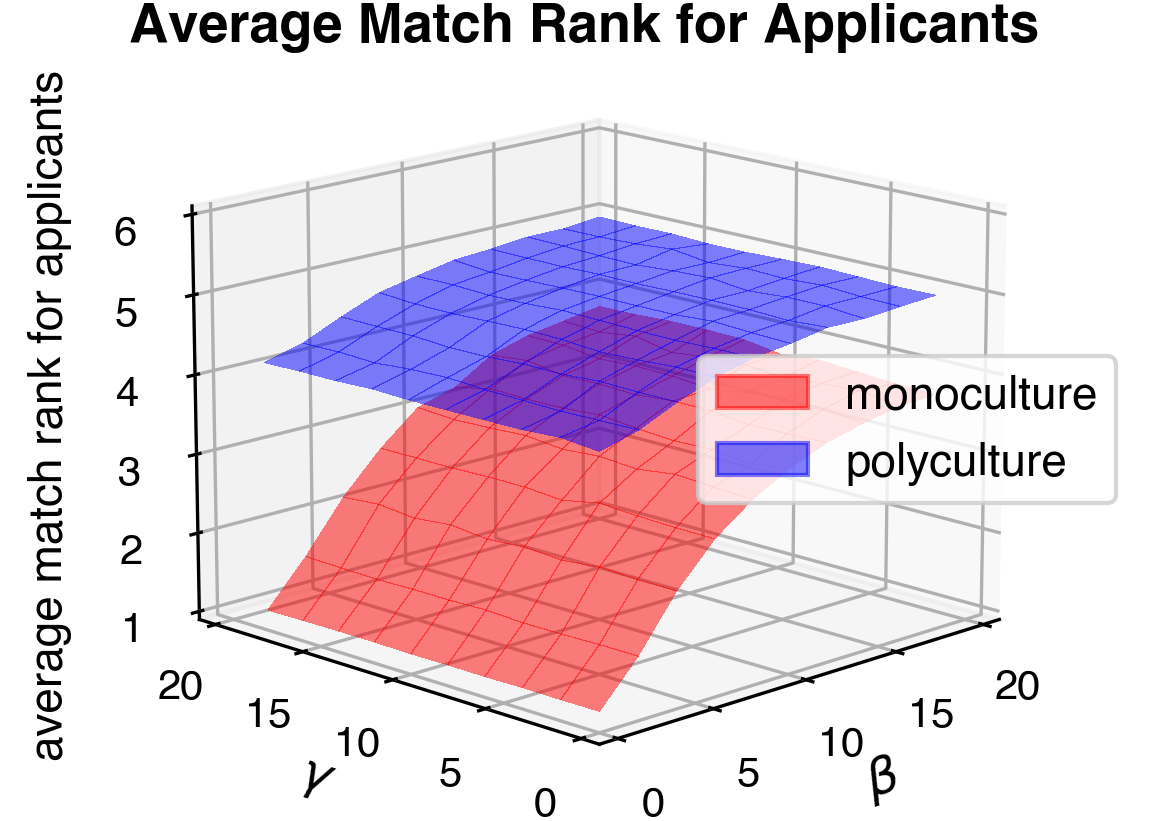}}}
    \caption{The average percentile value of matched applicants (left) and the average rank of firms matched to applicants (right), as a function of $\beta$ and $\gamma$, which control the level of global and local correlation between applicant preferences. (The case $\beta=\gamma=0$ corresponds to our theoretical setup.) For all choices of $\beta$ and $\gamma$ we consider here, average firm welfare is higher under polyculture while average applicant welfare is higher under monoculture (note that lower corresponds to a better outcome in the right plot).}
    \label{fig:exp-firm-applicant-welfare}
\end{figure}

\subsection{Firm welfare}

We test our first prediction: that polyculture selects more-preferred applicants in comparison to monoculture. To do this, we compute the average percentile value of applicants matched to firms, so that a higher average percentile corresponds to higher firm welfare. As shown in \Cref{fig:exp-firm-applicant-welfare} (left), average firm welfare is higher under polyculture across all $\beta$ and $\gamma$ we consider. In fact, our experiments suggest that firm welfare in both monoculture and polyculture remains consistent regardless of the correlation structure we choose.

\subsection{Applicant welfare}

We now test our second prediction: that monoculture yields higher total applicant welfare. Here, we measure total applicant welfare as the average rank of applicant matches---conditional on matching with a firm, what is expected rank of an applicant's match according to their preference list? A lower average rank corresponds to better applicant outcomes, since this means that applicants are on average matched to more-preferred options. As shown in \Cref{fig:exp-firm-applicant-welfare} (right), the average rank of applicant matches is worse under polyculture than monoculture for all $\beta$ and $\gamma$ we consider. Notice that for both monoculture and polyculture, the average rank of applicant matches increases as $\beta$ increases. Intuitively, this is true because the increased correlation between applicant preferences means that it is harder for all applicants to receive their preferred options.

We further test more specific predictions from \Cref{thm:top-choice}, that an applicant is more likely to be matched to their top choice under monoculture (\Cref{thm:top-choice}(i)), and that some applicant's are more likely to be matched overall under polyculture (\Cref{thm:top-choice}(iii)), leading to a non-stochastic-dominance result. Our experiments are plotted in \Cref{fig:exp-prob-match}. We vary the level of global correlation $\beta$ and consider the probability of matching conditional on an applicant's true value. Our experiments confirm the two predictions. Note that as $\beta$ increases, the probability of an applicant matching to their top choice is smaller both for monoculture and polyculture. Since many applicants share the same top choice, few applicants can be matched to their top choice.

\begin{figure}
    \centering
    \subfloat{{\includegraphics[width=0.49\linewidth]{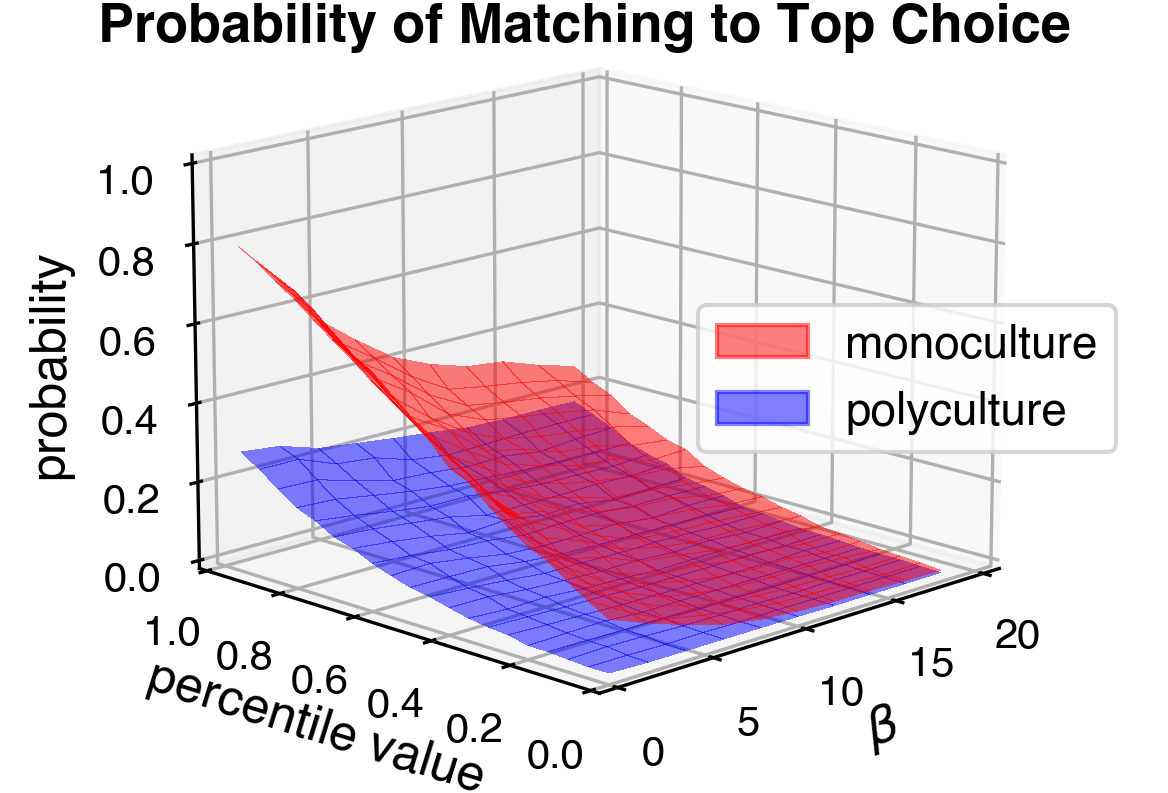}}}
    \subfloat{{\includegraphics[width=0.49\linewidth]{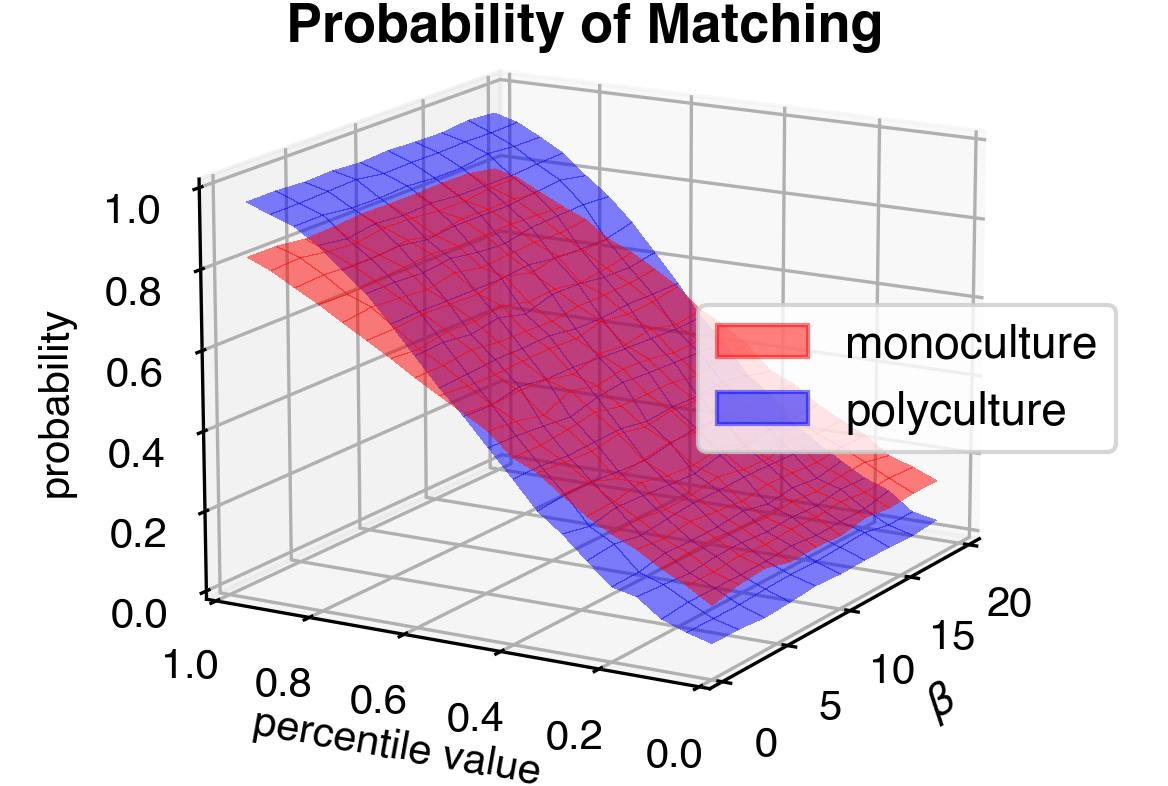}}}
    \caption{The probability an applicant matches to their top choice (right) and to any firm at all (right), as a function of the applicant's percentile true value and $\beta$ the level of global correlation in applicant preferences.}
    \label{fig:exp-prob-match}
\end{figure}

\subsection{Differential application access}

We now test our third main theoretical prediction: that monoculture is more robust to polyculture with respect to variation in the number of applications submitted by different applicants. Recall that in order to study differential access, it is necessary to specify an application strategy for each applicant---for an applicant that can apply to $k$ firms, which $k$ do they select? In our theoretical setup---where there was no correlation in applicant preferences---we benefited from the existence of a simple Nash equilibrium, where applicants all apply to their top $k$ choices. Given arbitrary correlation structures, the characterization of Nash equilibria remains an open question, and is an area of continued inquiry.\footnote{\cite{chade2006simultaneous} and \cite{ali2021college} consider the optimal choice of $k$ firms to apply to when admissions probabilities are independent and correlated, respectively. Neither, however, consider interactions between the strategies of multiple applicants. \cite{haeringer2009constrained} show the existence of Nash equilibria in the deferred acceptance algorithm (in addition to other mechanisms) when applicants can only apply to $k$ firms. However, it does not consider the stochasticity present in our model.}

Therefore, we implement two heuristic strategies: applying to one's $k$ most-preferred firms, and applying to $k$ randomly selected firms (in the order of the applicant's true preferences). The former, as we showed in \Cref{prop:nash}, is a Nash equilibrium when applicant preferences are fully uncorrelated (i.e., when $\beta = \gamma = 0).$ Intuitively, this strategy performs poorly when preferences are correlated: applicants who have low value and who can only apply to one firm should not apply to their most preferred firm, which is likely to be highly competitive. In this case, the second heuristic strategy we propose, applying to randomly selected firms (in the true order of the applicant's preferences) is likely to perform better. Roughly, applying randomly implements the commonly-used ``reach-match-safety'' approach, which has also been theoretically justified by \cite{ali2021college}.

In our experiments, we consider $\kappa$ to be the uniform distribution on $\{1,2,\cdots,10\}$, so applicants are allowed to apply to a random number of firms between $1$ and $10$. In \Cref{fig:exp-diff-access}, we plot the difference in the probability of matching between applicants who can apply to between $6$ and $10$ firms and applicants who can apply to between $1$ and $5$ firms. This is a measure of the benefit accrued by an applicant who applies to more firms. A larger difference is equivalent to less robustness to differential application access, since this implies that applicants who apply to more firms have a larger chance of being matched. As shown in \Cref{fig:exp-diff-access}, this difference is larger under polyculture than under monoculture for all choices of $\beta$ and $\gamma$ we consider. This holds both when applicants apply to their most-preferred firms and random firms.

\begin{figure}
    \centering
    \includegraphics[width=\linewidth]{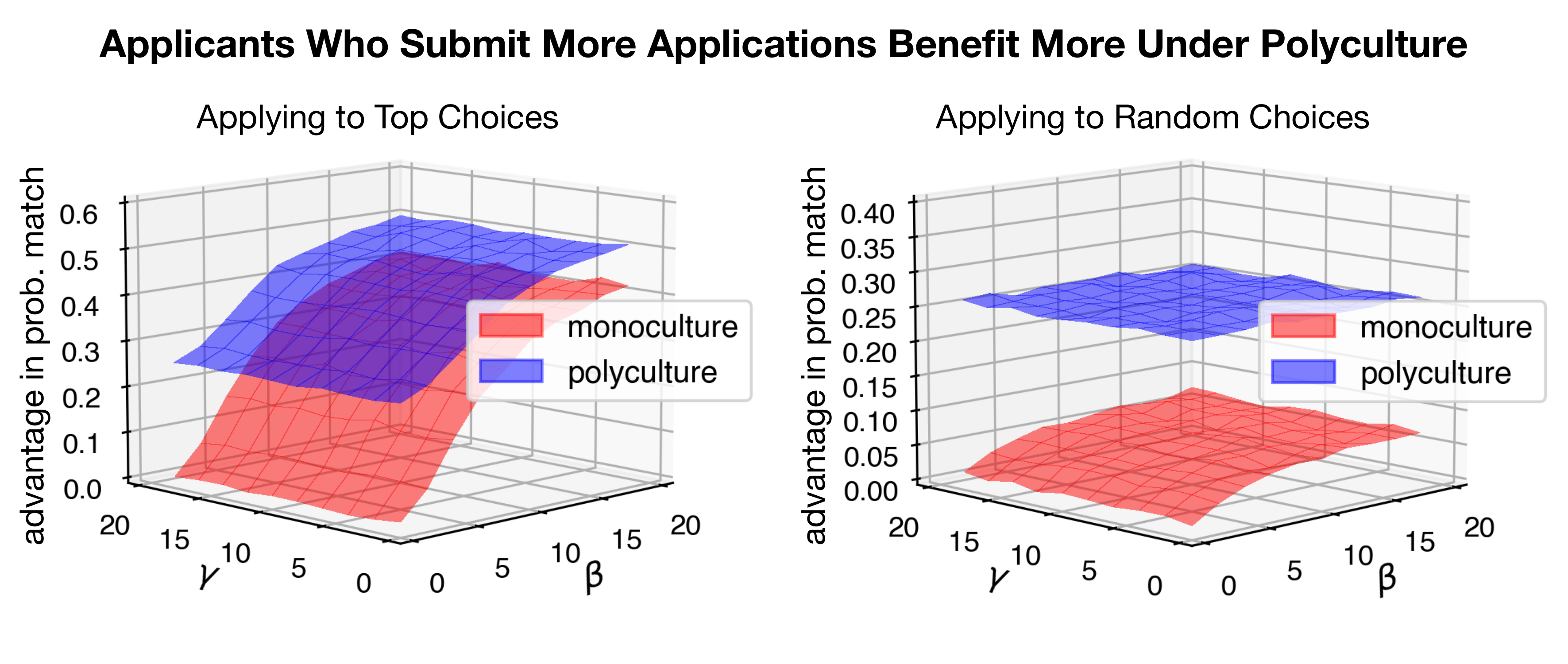}
    \caption{Difference in probability of matching among applicants who can apply to between 6 and 10 firms, and those who can apply to between 1 and 5 firms. This difference is always higher in polyculture than in monoculture across the parameters we consider. This holds both when applicants apply to their top choices (left) and when they apply to random firms in their order of preference (right).}
    \label{fig:exp-diff-access}
\end{figure}

We further test our prediction that firm welfare is affected more by differential application access under polyculture than under monoculture. To do this, we consider the change in firm welfare when moving from ``uniform application access'' where all applicants apply to all firms to differential application access according to $\kappa$. As shown in \Cref{fig:exp-diff-access-college-welfare}, this change is more significant under polyculture than under monoculture, for both application strategies. This supports our prediction.

\begin{figure}
    \includegraphics[width=\linewidth]{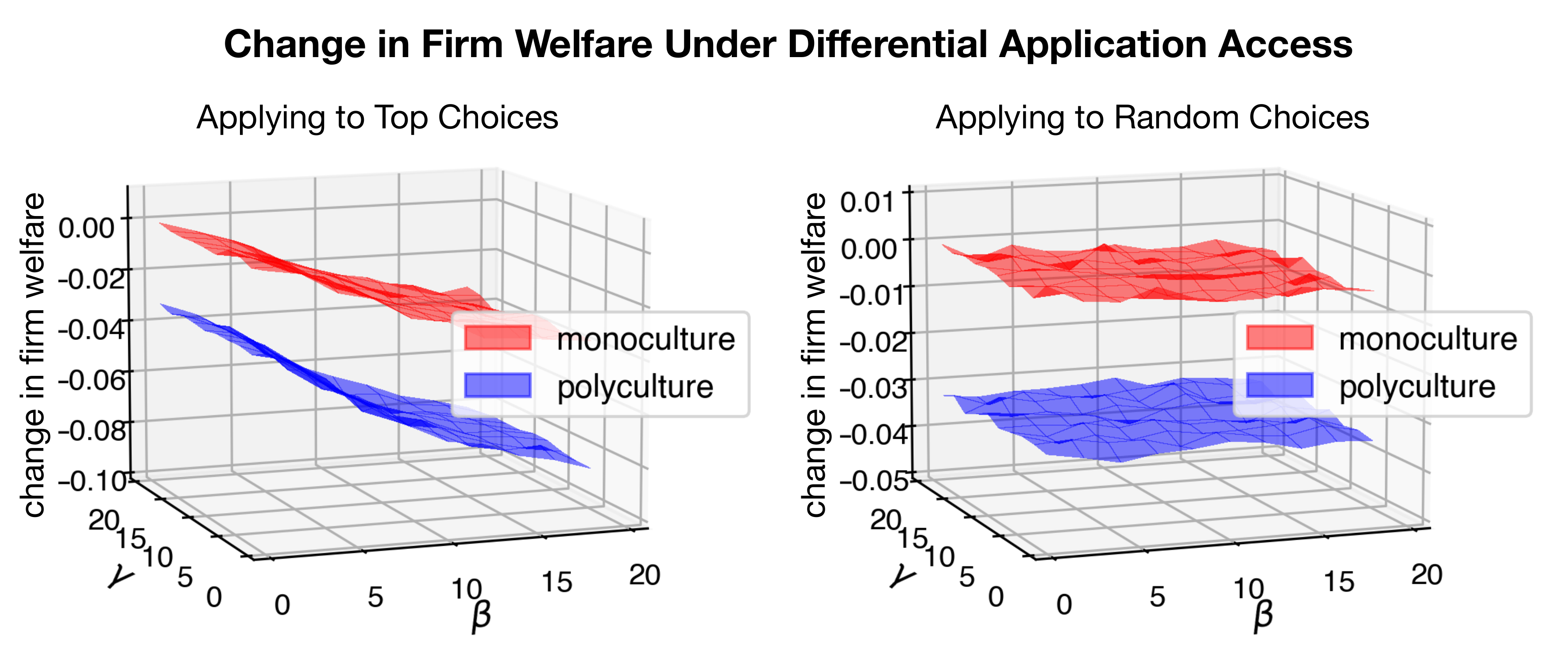}
    \caption{The change in college welfare when moving from (1) a market with no differential application access, to (2) a market with differential application access. Note that college welfare decreases more under polyculture. This holds both when applicants apply to their top choices (left) and when they apply to random firms in their order of preference (right).}
    \label{fig:exp-diff-access-college-welfare}
\end{figure}
\section{Extended Related Work}
\label{sec:extended_related}
Here, we further detail our relationship to the algorithmic monoculture and matching markets literature. We conclude with a discussion of how our work bridges these two lines of work, and reconciles some competing intuitions. 

\paragraph{Algorithmic monoculture.}

Closest to our work is that of \cite{kleinberg2021algorithmic}, who introduce a model of algorithmic monoculture in which there are two decision-makers who each hire one applicant. As in our model, decision-makers have shared true preferences over a set of applicants, instantiated as a ranked list. However, each decision-maker only has access to a noisy version of this ranked list: they each can either use an idiosyncratic noisy list or a common (but more accurate) noisy list. One decision-maker then selects the top applicant on their list. The second decision-maker chooses the top applicant on their list that has not yet been chosen. Welfare is measured by the rank of the chosen applicants according to the true preference list. \cite{kleinberg2021algorithmic} show that the decision-makers can be better off if they both use their idiosyncratic (more noisy) lists in comparison to using the shared list. In our work, we show that this effect becomes more pronounced when considering a large market with many decision-makers and applicants. By building on the model of \cite{azevedo2016supply}, we overcome issues of tractability, while also gaining a cleaner interpretation for why the effect found in \cite{kleinberg2021algorithmic} holds, in terms of distributional order statistics and their effect on admission thresholds for each decision-maker. Moreover, we are able to model two-sided effects, where applicants also have preferences over decision-makers. Absent in our work is the game-theoretic analysis of \cite{kleinberg2021algorithmic}, who show that there are cases where both decision-makers use the shared list in the Nash equilibrium, even when both would be better off using idiosyncratic lists. %

While \cite{kleinberg2021algorithmic} focus on the welfare of decision-makers, \cite{bommasani2022picking} consider the effect of algorithmic monoculture on applicants. In their setup, applicants are subject to binary decisions from several decision-makers. They define the rate of systemic failure to be the probability that a randomly drawn applicant receives a negative decision from all decision-makers. In computational experiments, they provide evidence that when decision-makers train machine learning models using the same data, the rate of systemic failure tends to increase. \cite{toups2023ecosystem} document homogeneous outcomes in several commercial APIs, which they suggest could be a consequence of algorithmic monoculture arising from shared data or model usage. These empirical studies do not consider market effects such as fixed capacity constraints of decision-makers; in the setup we consider, since decision-makers have fixed capacity, the number of applicants who experience systemic failure remains constant (assuming that applicants too only ``accept'' one offer).

Indeed, decision-maker capacity constraints play an important part in our results. \cite{jagadeesan2023improved} consider a different structure, in which several ``model-providers'' compete over users. In their setting, a user (noisily) prefers model-providers that produce a more accurate prediction for the user. A key difference between this model and ours is that the analog to the decision-maker (the model provider) does not have limited capacity; rather, a single model-provider can serve an unlimited number of users. Moreover, the model provider does not seek to evaluate users, but rather to make predictions in their service. \cite{jagadeesan2023improved} shows that as the ability of model-providers to produce more accurate predictions increases overall, the social welfare of users can decrease as a consequence of growing homogeneity in predictions. (Users benefit from diversity since it means that there is likely at least one model provider that is accurate for them.) We also note that in \cite{jagadeesan2023improved}'s model, monoculture arises endogenously, while monoculture is specified exogenously in our model.

An emerging line of work also considers normative aspects of monoculture. \cite{creel2022algorithmic} argue that \textit{arbitrariness} in algorithmic decision-making is not generally of moral concern---however, they argue, arbitrary decision-making \textit{at scale} is of concern, since it can result in \textit{systemic exclusion}: ``Exclusion from a broad swath of opportunity in an important sector of life is likely to be morally problematic.'' Our work challenges the implicit assumption that the widespread adoption of a common algorithm necessarily increases the number of individuals excluded from all opportunities. When the number of ``opportunities'' (positions at a firm, seats at a college) is constant, the number of individuals that receive these opportunities (and who are denied from all opportunities) should remain constant. Of course, not all domains of algorithmic decision-making share this feature. In the domains that do, our work suggests systemic exclusion is not clearly the primary concern, but rather \textit{who} faces systemic exclusion, as well as the other measures of welfare (e.g., worker power). \cite{jain2023algorithmic} argue that algorithmic decision-making should be viewed through the lens of \textit{opportunity pluralism} \citep{fishkin2014bottlenecks}. A key implication of their argument is that the use of algorithmic decision-making by one or multiple decision-makers should be analyzed with respect to its impact on how opportunities are restricted at large. Our work contributes to this type of analysis, helping develop a more precise understanding of the effects of algorithmic decision-making at a system level.

\paragraph{Matching markets.}
The model we introduce here builds on the broader analysis of matching markets. In particular, a substantial literature analyzes the outcomes of random matching markets. Theoretical results have established the expected number of stable matchings (i.e., the size of the core) as well as the distribution of outcomes for each participant (e.g., \cite{pittel1989average, immorlica2003marriage, ashlagi2017unbalanced}). 

We depart from this literature by modeling the process of \textit{preference approximation}. Whereas existing analysis focuses on a single set of stated preferences, we assume that participants (specifically, decision-makers) can only obtain or communicate an \textit{approximation} of their own preferences. As a result, stable matchings are computed with respect to these approximate preferences. We then analyze the matching outcome with respect to the \textit{true} preferences of participants. Our model can thus be viewed as having two components:
\begin{itemize}
    \item First, a \textit{preference approximation} model, in which each participant acquires approximate knowledge of their true preferences over the other side.\footnote{Note that in our model, only the preferences of decision-makers are ``approximated,'' while applicants are assumed to know their own true preferences. This preference approximation further relates to recent literature in modeling the role of standardized testing in admissions, where tests (among other features) help schools estimate their value for each applicant \cite{emelianov2022fair,gargtesting21,liuoptionaltests21,niumultipletests22}. However, studying more than two schools in this framework has proven theoretically challenging. \cite{castera2022statistical} analyzes a related model to understand the effect of statistical discrimination in matching; they too focus on a two school setting. A broader insight of our analysis is that the limiting regime in which the number of schools grows large can be more tractable than the strictly finite setting.}
    \item Second, a \textit{stable matching} model, which takes as input the approximate preferences of participants, and outputs a stable matching with respect to these preferences.
\end{itemize}
For the stable matching model, we adopt a model with a continuum of students \citep{abdulkadirouglu2015expanding, azevedo2016supply}. This allows for an especially tractable analysis. A limitation of this approach is that we do not capture the effects of stochasticity in small markets (particularly, where the capacity of decision-makers is small). In our framework, it is possible to substitute a stable matching model that is more amenable to such analysis, such as the standard finite model or that of \cite{arnosti2022continuum}.

An emerging line of work similarly considers matching markets (both centralized and decentralized) in which participants are unsure of their true preferences or broader information about the market. Existing work primarily focuses on settings where participants gain information sequentially (e.g., \cite{liu2020competing, liu2021bandit, dai2021learning, ionescu2021incentives, jeloudar2021decentralized, immorlica2020information}). Our work focuses instead on a one-shot setting, where we analyze the outcomes of matchings that arise given the current information of each participant. 

Closest to our work is a model introduced by \cite{castera2022statistical} to analyze statistical discrimination\footnote{See \cite{fang2011theories} for a survey of the statistical discrimination literature.} in stable matching. In their model, the values of applicants from two groups are estimated by two decision-makers. Like us, they leverage a continuum model of applicants. Their findings focus on aggregate measures of applicant-side welfare across groups. They show, for example, that the group of applicants evaluated with less correlation is matched with higher probability. More directly related to our work, they show that decreased correlation (even in one group), reduces the number of applicants who receive their top choice; this parallels our finding in \Cref{thm:top-choice}(ii), which implies that fewer applicants are matched to their top choice under polyculture (i.e., when there is total correlation). This finding reflects a general result developed in the school choice literature: that increasing correlation in how decision-makers evaluate applicants also increases the total number of applicants matched to their top choice. We discuss this line of literature next.

\paragraph{Tie-breaking rules in school choice.}
Our work has close connections to the matching markets literature studying tie-breaking rules in the context of school choice, where lottery numbers are often used to determine the school preferences inputted into stable matching mechanisms \citep{abdulkadirouglu2003school}. A significant line of empirical and theoretical work compares the use of \textit{single tie-breaking} (STB) where each student is assigned a single lottery number across all schools, and \textit{multiple tie-breaking} (MTB) where each student is assigned a different random lottery number for each school. 

A main finding across this literature is that STB results in more students being assigned to their top choice school (e.g., \cite{abdulkadirouglu2009strategy, de2023performance, ashlagi2019assigning, arnosti2023lottery, allman2023rank}). This finding is echoed in our \Cref{thm:top-choice}. Indeed, one may directly compare monoculture to STB, since each applicant in our model has a single estimated value used by all decision-makers. From the perspective of population-level statistics of overall student welfare, monoculture and STB are both equivalent to a model in which all firms (schools) have the same preferences over applicants. At a high level, both sets of results rely on a similar intuition: in settings with two-sided preferences, applicants as a whole benefit if \textit{their} preferences most influence matches; this happens when schools share preferences (as induced by either monoculture or STB, but not by polyculture or MTB).

Our setup departs from the tie-breaking literature by having each applicant's estimate depend on the applicant's \textit{true value}. Thus, the ``noise'' in their estimate (the equivalent of a tie-breaking score) is not just differentiating between equivalent applicants; it can move a lower-value applicant over a higher-value one. As a result, our analysis further focuses on comparing monoculture and polyculture from the perspective of an \textit{individual} applicant with a particular value. This also allows us to analyze welfare from the perspective of the firm (who prefer to be matched with applicants of higher true value). Polyculture and MTB also share differences even at the population level; in polyculture, there remains some correlation between ``lottery numbers'' across schools (firms) depending on the applicant's true value.

In addition to showing that STB results in more students being matched to their top choice,\footnote{In fact, \cite{arnosti2023lottery} shows more precise results, characterizing the probability that a student is matched to one of their top $k$ choices.} \cite{arnosti2023lottery} shows that students who list many schools are more likely to be matched under MTB, while students who list few schools are more likely to be matched under STB. This accords with our finding in \Cref{thm:diff-app-access} that applicants who submit more applications gain an advantage under polyculture but not under monoculture. We show that this further implies that (1) students with lower true values but who submit more applications can be matched over students with higher true values but who submit fewer applications, and that (2) differences in the number of applications submitted by students reduces firm welfare under polyculture, but not under monoculture.

We anticipate that the techniques developed in the study of tie-breaking rules may be useful for analyzing outcomes in algorithmic monoculture, especially on the applicant side. The introduction of heterogeneity on the applicant side adds a dimension of analysis beyond what is considered by the existing tie-breaking literature, and which is ripe for further inquiry.

\paragraph{Bridging the literatures.} Our work bridges two lines of work: algorithmic monoculture and matching markets. We briefly elaborate on the nature of this connection. The initial inquiry into algorithmic monoculture has been limited by a lack of market-level analysis. The choice of a firm to adopt a certain hiring algorithm not only affects outcomes at that firm, but also other firms. The nature of this effect is further dictated by the preferences of applicants. We have shown how machinery in the matching markets literature is well-equipped to handle these effects. Adopting this machinery enhances our understanding of algorithmic monoculture. While the algorithmic monoculture literature has largely pointed to negative effects of monoculture, both on the firm and applicant side, we showed that monoculture can have positive welfare effects for applicants---a competing intuition that can be found in the matching markets literature on tie-breaking, which has found that correlation in school-side preferences can benefit students in a school choice setting. In this way, we anticipate that the matching markets literature can provide further insight into the emerging study of algorithmic monoculture.

At the same time, our work expands the scope of questions that can be analyzed using a matching markets framework. While existing work in matching markets largely takes preferences at face value, we further model a stage of \textit{preference approximation}, in which participants form noisy approximations of their true preferences. In this work, we focused on correlation in this approximation process, but the question can be posed more generally. Our results suggest that the structure of the approximation process can have intriguing consequences on the quality of the resulting matching.

\section{Conclusion}

In this paper, we studied the effects of algorithmic monoculture by employing a matching markets model. This allowed us to capture a key feature of domains that may be affected by monoculture: competition between and the limited capacity of decision-makers (firms, colleges). By incorporating these features, and by considering large markets with many participants, we strengthened, challenged, and expanded upon the existing understanding of monoculture. We presented three main theoretical findings, which we verified in computational experiments. All else equal, monoculture
\begin{enumerate}
    \item selects less-preferred applicants,
    \item yields higher overall applicant welfare, though effects for individual applicants depend on their value for decision-makers and risk tolerance, 
    \item is more robust to differences in the number of applications submitted by different applicants.
\end{enumerate}
The framework we have introduced here inherits the tractability of a recent matching markets model (that of \cite{azevedo2016supply}), but we expect is also flexible to other approaches. At a higher level, the framework we have taken here can be described as follows: consider a market where participants have \textit{true} preferences, specify how preferences are realized in practice as \textit{approximate} preferences, and then studying the market outcome (here, the stable matching with respect to the approximate preferences) with respect to participants' true preferences. We hope that this approach can be adopted to study a broader range of questions---in addition to furthering our understanding of algorithmic monoculture.

\newpage
{\small
\bibliography{bib}
}

\newpage
\appendix
\section{Preliminary Results}

\subsection{Observations About Measures with Connected Support} 

To aid subsequent proofs, we recall our assumption on the probability measure $\eta$ (of student values) and the probability distribution $\DD$ (the noise distribution); namely, that they have connected support. We then make two useful observations.

Let $\pi$ be the probability measure associated with $\DD$. Then we make the assumption that both $\eta$ and $\pi$ have \textit{connected support}, where we mean that the smallest closed set of measure $1$ is an interval. Let these respective intervals be $[V_-, V_+]$ and $[X_-, X_+]$. (Note that $V_-, X_-$ may be equal to $-\infty$ and $V_+, X_+$ may be equal to $\infty$.) In what follows we let $V$ and $X$ be random variables distributed according to the probability measures $\eta$ and $\pi$, and let $F_V$ and $F_X$ denote their cdfs.

\begin{proposition}\label{prop:cdf-increasing}
The following hold:
\begin{itemize}
    \item[(i)] $F_V$ is strictly increasing on the interval $(V_-, V_+),$ and $F_V(v)\in (0,1)$ when $v\in (V_-, V_+).$
    \item[(ii)] $F_X$ is strictly increasing on the interval $(X_-, X_+),$ and $F_X(x)\in (0,1)$ when $x\in (X_-, X_+).$
    \item[(iii)] For all integers $\numcolleges\ge 1,$ $F_X^\numcolleges$ is strictly increasing on the interval $(X_-, X_+),$ and $F_X^\numcolleges(x)\in (0,1)$ when $x\in (X_-, X_+).$
\end{itemize}
\end{proposition}

\begin{proof}
We show the result for $F_V$; the result for $F_X$ is (clearly) analogous. 

First suppose that $F_V$ were not strictly increasing on $V_-, V_+$, such that there exists $v_1, v_2$ such that $V_- < v_1 < v_2 < V_+$ and $F_V(v_1) = F_V(v_2)$. Then $\eta((v_1, v_2)) = F_V(v_2) - F_V(v_1) = 0.$ This implies that $\eta([V_-, V_+]\setminus [v_1, v_2]) = 1 - 0 = 1,$ which contradicts the assumption that $[V_-, V_+]$ is the smallest closed set of $\eta$-measure $1$.

To show that $F_V(v)\in (0,1)$ when $v\in (V_-, V_+)$, it suffices to show that $F_V(V_-)=0$ and $F_V(V_+)=1.$ Assume otherwise, such that $F_V(V_-) > 0$ or $F_V(V_+) < 1.$ This would imply that $\eta([V_-, V_+]) < 1$, which gives a contradiction.
\end{proof}

\begin{proposition}\label{prop:nonzero-measure}
    Let $I$ be an interval such that
    \begin{equation}
        I \cap (V_-, V_+) \neq \emptyset.
    \end{equation}
    Then $\eta(I) > 0.$
\end{proposition}

\begin{proof}
    There exists $a,b$ such that $V_-\le a<b\le V_+$ such that $(a,b)\subseteq I \cap (V_-, V_+)$. Therefore, $\eta(I) > \eta((a,b)).$ We have that $\eta((a,b)) = F_V(b) - F_V(a) > 0$, where we used that $F_V$ is strictly increasing on $[V_-, V_+]$, as shown in \Cref{prop:cdf-increasing}. This gives the result.
\end{proof}

\subsection{The Lattice Theorem for Stable Matchings}
A powerful result in the stable matching literature shows that stable matchings have a complete lattice structure. In this section, we state this result in our setup (Theorem A.1 in \cite{azevedo2016supply}). Consider the lattice operators $\vee$ and $\wedge$ on cutoffs in $\RR^C$, where for a set $Z$ of cutoffs,
\begin{equation}
    \left(\bigvee_{P\in Z} P\right)_\college = \sup_{P\in Z} P_\college
\end{equation}
and
\begin{equation}
    \left(\bigwedge_{P\in Z} P\right)_\college = \inf_{P\in Z} P_\college.
\end{equation}

\begin{proposition}[The Lattice Theorem]\label{prop:lattice}
The set of market-clearing cutoffs forms a complete lattice with respect to $\vee$ and $\wedge$ as defined above.
\end{proposition}

For our purposes, this will be useful to establish the two ``Equal Cutoffs Lemmas'' we prove (\Cref{prop:equal-cutoffs} and \Cref{lem:equal-cutoffs-diff-access}), in which we show that our economies induce unique stable matchings, in which all cutoffs are equal. The high-level idea is to show that the maximum and minimum cutoff vectors (1) are each comprised of identical cutoffs, and (2) coincide.

\section{Proofs for \Cref{sec:results}}

We begin by proving the Equal Cutoffs Lemma for $\theta_{\mono}$ and $\theta_{\poly}.$

\begin{proof}[\textbf{Proof of \Cref{prop:equal-cutoffs}}]
    For now, consider a generic $\theta\in \{\theta_{\mono} ,\theta_{\poly}\}$. Let $Z$ be the set of market-clearing cutoffs. Then by symmetry, if $P = (P_1, \cdots, P_\numcolleges) \in Z$, then any permutation of $P$ is also in $Z$. It follows that
    \begin{equation}
        \sup_{P\in Z} P_1 = \sup_{P\in Z} P_2 = \cdots = \sup_{P\in Z} P_\numcolleges
    \end{equation}
    and
    \begin{equation}
        \inf_{P\in Z} P_1 = \inf_{P\in Z} P_2 = \cdots = \inf_{P\in Z} P_\numcolleges.
    \end{equation}
    Call these two common values $P_+$ and $P_-$ respectively. By \Cref{prop:lattice}, $(P_+, \cdots, P_+), (P_-, \cdots, P_-)\in Z$. These are the greatest and least elements of the complete lattice of market-clearing cutoffs. Therefore, it suffices to show that $P_+ = P_-.$ We show this separately for $\theta = \theta_{\mono}$ and $\theta = \theta_{\poly}.$

    \paragraph{Case 1: Monoculture.} Consider $\theta = \theta_{\mono}.$ It suffices to show that $P_+ = P_-$. Assume for sake of contradiction that $P_+ > P_-$. (By definition, $P_+\ge P_-$.) By the definition of market clearing, the total measure of students matched is equal to $S$ under either $(P_-, \cdots, P_-)$ or $(P_+, \cdots, P_+)$. Therefore,
    \begin{equation}
        \int_\RR \Pr[v + X > P_+]\,d\eta(v) = \totalsupply = \int_\RR \Pr[v + X > P_-]\,d\eta(v).
    \end{equation}
    To give the desired contradiction, we show that
    \begin{equation}
        \int_\RR \Pr[v + X > P_+]\,d\eta(v) < \int_\RR \Pr[v + X > P_-]\,d\eta(v).
    \end{equation}
    Clearly, $\Pr[v + X > P_+] \le \Pr[v + X > P_-].$ Therefore, it suffices to show that $\Pr[v + X > P_+] < \Pr[v + X > P_-]$ on a set $I$ of positive $\eta$-measure. We show this is the case for $I = (P_- - X_+, P_- - X_-).$
    
    We first show that the inequality is satisfied on $I$. Note that for $v\in I,$ we have that $P_- - v \in (X_-, X_+).$ Therefore,
    \begin{align}
        \Pr[v + X > P_-] - \Pr[v + X > P_+] = F_X(P_+ - v) - F_X(P_- - v) > 0
    \end{align}
    where the last inequality follows from observing that by $P_+-v > P_- - v$ by assumption and applying \Cref{prop:cdf-increasing}, which says that $F_X$ is strictly increasing on $(X_-, X_+)$.

    It remains to show that the interval $I$ has positive $\eta$-measure. We begin by showing that
    $P_- > V_- + X_-.$
    This is true because for $P \le V_- + X_-$,
    \begin{equation}
        \int_\RR \Pr[v + X > P]\,d\eta(v) \ge \int_\RR \Pr[v + X > V_- + X_-]\,d\eta(v) = \int_\RR 1\,d\eta(v) = 1,
    \end{equation}
    contradicting the assumption that $S < 1$.
    Similarly, $P_- < V_+ + X_+.$ It follows that $P_- - X_+ < V_+$ and $P_- - X_- > V_-.$
    Therefore, the interval $I = (P_- - X_+, P_- - X_-)$ intersects $(V_-, V_+)$. So \Cref{prop:nonzero-measure} implies that it has positive measure.

    \paragraph{Case 2: Polyculture.} The result holds analogously for $\theta = \theta_{\poly}$, where we replace $X$ with $X^{(\numcolleges)}$ (and thus $F_X$ with $F_X^\numcolleges$).

    It suffices to show that $P_+ = P_-$. Assume for sake of contradiction that $P_+ > P_-$. (By definition, $P_+\ge P_-$.) By the definition of market clearing, the total measure of students matched is equal to $S$ under either $(P_-, \cdots, P_-)$ or $(P_+, \cdots, P_+)$. Therefore,
    \begin{equation}
        \int_\RR \Pr[v + X^{(\numcolleges)} > P_+]\,d\eta(v) = \totalsupply = \int_\RR \Pr[v + X^{(\numcolleges)} > P_-]\,d\eta(v).
    \end{equation}
    To give the desired contradiction, we show that
    \begin{equation}
        \int_\RR \Pr[v + X^{(\numcolleges)} > P_+]\,d\eta(v) < \int_\RR \Pr[v + X^{(\numcolleges)} > P_-]\,d\eta(v).
    \end{equation}
    Clearly, $\Pr[v + X^{(\numcolleges)} > P_+] \le \Pr[v + X^{(\numcolleges)} > P_-].$ Therefore, it suffices to show that $\Pr[v + X^{(\numcolleges)} > P_+] < \Pr[v + X^{(\numcolleges)} > P_-]$ on a set $I$ of positive $\eta$-measure. We show this is the case for $I = (P_- - X_+, P_- - X_-).$
    
    We first show that the inequality is satisfied on $I$. Note that for $v\in I,$ we have that $P_- - v \in (X_-, X_+).$ Therefore,
    \begin{align}
        \Pr[v + X^{(\numcolleges)} > P_-] - \Pr[v + X^{(\numcolleges)} > P_+] = F_X^\numcolleges(P_+ - v) - F_X^\numcolleges(P_- - v) > 0
    \end{align}
    where the last inequality follows from observing that by $P_+-v > P_- - v$ by assumption and applying \Cref{prop:cdf-increasing}, which says that $F_X^\numcolleges$ is strictly increasing on $(X_-, X_+)$.

    It remains to show that the interval $I$ has positive $\eta$-measure. We begin by showing that
    $P_- > V_- + X_-.$
    This is true because for $P \le V_- + X_-$,
    \begin{equation}
        \int_\RR \Pr[v + X^{(\numcolleges)} > P]\,d\eta(v) \ge \int_\RR \Pr[v + X^{(\numcolleges)} > V_- + X_-]\,d\eta(v) = \int_\RR 1\,d\eta(v) = 1,
    \end{equation}
    contradicting the assumption that $\totalsupply < 1$.
    Similarly, $P_- < V_+ + X_+.$ It follows that $P_- - X_+ < V_+$ and $P_- - X_- > V_-.$
    Therefore, the interval $I = (P_- - X_+, P_- - X_-)$ intersects $(V_-, V_+)$. So \Cref{prop:nonzero-measure} implies that it has positive measure.
\end{proof}

As a consequence of the above proof,
\begin{equation}\label{eq:interval-mono}
    \eta((P_{\mono} - X_+, P_{\mono} - X_-)) > 0
\end{equation}
and
\begin{equation}\label{eq:interval-poly}
    \eta((P_{\poly} - X_+, P_{\poly} - X_-)) > 0,
\end{equation}
which will be useful in subsequent proofs.

\begin{proof}[\textbf{Proof of \Cref{cor:cutoff-inequality}}]
    Noting that the total measure of students matched is the same under both monoculture and polyculture (equal to $\totalsupply$), we have, using \eqref{eq:match-mono} and \eqref{eq:match-poly} that
    \begin{equation}
        \int_{\RR} \Pr[v + X > P_{\mono}]\,d\eta(v) = \totalsupply = \int_{\RR} \Pr[v + X^{(\numcolleges)} > P_{\poly}]\,d\eta(v).
    \end{equation}
    Assume for the sake of contradiction that $P_{\mono}\ge P_{\poly}.$ To give the desired contradiction, we will show that this implies
    \begin{equation}
        \int_{\RR} \Pr[v + X > P_{\mono}]\,d\eta(v) < \int_{\RR} \Pr[v + X^{(\numcolleges)} > P_{\poly}]\,d\eta(v).
    \end{equation}
    Clearly,
    \begin{equation}
        \Pr[v + X > P_{\mono}] \le \Pr[v + X^{(\numcolleges)} > P_{\poly}] 
    \end{equation}
    for all $v$. We now show that
    \begin{equation}
        \Pr[v + X > P_{\mono}] < \Pr[v + X^{(\numcolleges)} > P_{\poly}] 
    \end{equation}
    for $v\in (P_{\poly}-X_+, P_{\poly}-X_-)$, an interval which we will show has positive $\eta$-measure. We have that
    \begin{equation}
        \Pr[v + X > P_{\mono}] \le \Pr[v + X > P_{\poly}],
    \end{equation}
    so it suffices to show that 
    \begin{equation}\label{eq:water}
        \Pr[v + X > P_{\poly}] < \Pr[v + X^{(\numcolleges)} > P_{\poly}].
    \end{equation}
    We have that
    \begin{align}
        \Pr[v + X > P_{\poly}] &= 1 - F_X(P_{\poly} - v)\\
        \Pr[v + X^{(\numcolleges)} > P_{\poly}] &= 1 - F_X^\numcolleges(P_{\poly} - v).
    \end{align}
    Therefore, since $\numcolleges\ge 2$, \eqref{eq:water} holds whenever $F_X(P_{\poly} - v)\in (0,1)$, which, by \Cref{prop:cdf-increasing}, is true when $P_{\poly} - v \in (X_-, X_+).$ Equivalently, $v \in (P_{\poly} - X_+, P_{\poly} - X_-).$ This interval has positive measure \eqref{eq:interval-poly}.
\end{proof}

\subsection{Proof of \Cref{thm:wisdom}}
We first prove part (i) and then part (ii).

\subsubsection{\Cref{thm:wisdom}(i)} 
The proof of \Cref{thm:wisdom}(i) requires a few key observations, which we provide as separate (but closely related) parts in the following lemma.

\begin{lemma}\label{lem:potato}
If $\DD$ is maximum concentrating, the following hold for all $\epsilon, \delta > 0$ when $\numcolleges$ is sufficiently large:
\begin{enumerate}
    \item[(i)] For all $v < P_{\poly} - \EE\left[X^{(\numcolleges)}\right] - \frac{\delta}{3},$
    \begin{equation}
        \Pr[\mu_{\poly}(v)\in \numcolleges] < \epsilon.
    \end{equation}
    \item[(ii)] For all $v > P_{\poly} - \EE\left[X^{(\numcolleges)}\right] + \frac{\delta}{3},$
    \begin{equation}
        \Pr[\mu_{\poly}(v)\in \numcolleges] > 1 - \epsilon.
    \end{equation}
    \item[(iii)] 
    \begin{equation}
\left|\left(P_{\poly} - \EE\left[X^{(\numcolleges)}\right]\right) - v_\totalsupply\right| < \frac{2\delta}{3}.
    \end{equation}
\end{enumerate}
\end{lemma}

We provide some intuition for each part of the lemma, and how they come together to show \Cref{thm:wisdom}. First consider the value $P_{\poly} - \EE\left[X^{(\numcolleges)}\right],$ which is a natural ``threshold'' in the following sense: The expected maximum score of a student with value $P_{\poly} - \EE\left[X^{(\numcolleges)}\right]$ is exactly the cutoff $P_{\poly}.$ Part (i) of the lemma says that a student with even a slightly lower score than $P_{\poly} - \EE\left[X^{(\numcolleges)}\right]$ has a very low chance of having their maximum score exceed the cutoff $P_{\poly}.$ Meanwhile, part (ii) says that a student with even a slightly higher score than $P_{\poly} - \EE\left[X^{(\numcolleges)}\right]$ is almost certain to have their maximum score exceed $P_{\poly}.$ Part (iii) of the lemma says that this threshold $P_{\poly} - \EE\left[X^{(\numcolleges)}\right]$ cannot differ much from $v_\totalsupply$, which we recall is defined such that exactly an $S$ proportion of students have value higher than $v_\totalsupply$. Parts (i) and (ii) essentially show that there is a threshold at which the probability a student is matched jumps from nearly $0$ to almost $1$; part (iii) specifies the location of this threshold. In combination, these parts show \Cref{thm:wisdom}:

\begin{proof}[\textbf{Proof of \Cref{thm:wisdom} using \Cref{lem:potato}.}]
We first show that for $v < v_\totalsupply,$ for all $\epsilon > 0$, $\Pr[\mu_{\poly}(v)\in \colleges] < \epsilon$ for all $\numcolleges$ sufficiently large.
Consider any $v < v_\totalsupply$. Then there exists $\delta > 0$ such that $v < v_\totalsupply - \delta$. Using \Cref{lem:potato}(iii) and taking $\numcolleges$ sufficiently large, it follows from $v < v_\totalsupply - \delta$ that $v < P_{\poly} - \EE\left[X^{(\numcolleges)}\right] - \frac{\delta}{3}.$ In turn, again taking $\numcolleges$ sufficiently large, \Cref{lem:potato}(i) implies that
\begin{equation}
    \Pr[\mu_{\poly}(v)\in \colleges] < \epsilon.
\end{equation}

We next show that for $v > v_\totalsupply,$ for all $\epsilon > 0$, $\Pr[\mu_{\poly}(v)\in \colleges] > 1 - \epsilon$ for all $\numcolleges$ sufficiently large. Consider any $v > v_\totalsupply$. Then there exists $\delta > 0$ such that $v > v_\totalsupply + \delta$. Taking $\numcolleges$ sufficiently large, \Cref{lem:potato}(iii) shows that $v > v_\totalsupply + \delta$ implies that $v > P_{\poly} - \EE\left[X^{(\numcolleges)}\right] + \frac{\delta}{3},$ and \Cref{lem:potato}(ii) implies that
\begin{equation}
    \Pr[\mu_{\poly}(v)\in \colleges] > 1 - \epsilon.
\end{equation}
This completes the proof.
\end{proof}

\begin{proof}[\textbf{Proof of \Cref{lem:potato}}]
We show the three parts in order.

    We first show part (i). Suppose that $v < P_{\poly} - \EE\left[X^{(\numcolleges)}\right] - \frac{\delta}{3}$. Then, using \Cref{prop:prob-match}(ii),
    \begin{align}
        \Pr\left[\mu_{\poly}(v)\in \colleges\right] &= \Pr\left[v + \max_{\college\in \colleges}X_\college > P_{\poly}\right]\\
        &= \Pr\left[X^{(\numcolleges)} > P_{\poly} - v\right] \\
        &< \Pr\left[X^{(\numcolleges)} > \EE\left[X^{(\numcolleges)}\right] + \frac{\delta}{3}\right]\\
        &= \Pr\left[X^{(\numcolleges)} - \EE\left[X^{(\numcolleges)}\right] > \frac{\delta}{3}\right] \\
        &< \epsilon,
    \end{align}
    where the last line follows from the definition of $\DD$ being maximum concentrating and by taking $\numcolleges$ sufficiently large. 
    
    Part (ii) can be shown analogously.

    To show part (iii), we use parts (i) and (ii). For $\numcolleges$ sufficiently large,
    \begin{align}
        S &= \int_{-\infty}^\infty \Pr[\mu_{\poly}(v)\in \colleges] \,d\eta(v)\\
        &= \int_{-\infty}^{P_{\poly} - \EE\left[X^{(\numcolleges)}\right] - \frac{\delta}{3}} \Pr[\mu_{\poly}(v)\in \colleges] \,d\eta(v)
        + \int_{P_{\poly} - \EE\left[X^{(\numcolleges)}\right] - \frac{\delta}{3}}^\infty \Pr[\mu_{\poly}(v)\in \colleges] \,d\eta(v)\\
        &< \int_{-\infty}^{P_{\poly} - \EE\left[X^{(\numcolleges)}\right] - \frac{\delta}{3}} \epsilon \,d\eta(v)
        + \int_{P_{\poly} - \EE\left[X^{(\numcolleges)}\right] - \frac{\delta}{3}}^\infty 1 \,d\eta(v) \label{eq:lime-1} \\
        &= \eta\left(\left(-\infty, P_{\poly} - \EE\left[X^{(\numcolleges)}\right] - \frac{\delta}{3}\right)\right) \cdot \epsilon + \eta\left(\left(P_{\poly} - \EE\left[X^{(\numcolleges)}\right] - \frac{\delta}{3}, \infty \right)\right),\label{eq:lime-2}
    \end{align}
    where \eqref{eq:lime-1} follows from part (i) and by taking $\numcolleges$ sufficiently large.

    Now define $A := \eta\left(\left(P_{\poly} - \EE\left[X^{(\numcolleges)}\right] - \frac{\delta}{3}, \infty \right)\right).$ Substituting this notation into \eqref{eq:lime-2}, we have
    \begin{align}
        S &< (1 - A)\cdot \epsilon + A\\
        S - \epsilon &< (1 - \epsilon)A \\
        A &> \frac{S - \epsilon}{1 - \epsilon}.
    \end{align}
    Taking $\epsilon$ sufficiently small, $A$ is arbitrarily close to $S$, which implies---using that $\eta$ has connected support---that $P_{\poly} - \EE\left[X^{(\numcolleges)}\right] - \frac{\delta}{3} < v_\totalsupply + \delta'$ for any $\delta' > 0$ when $\numcolleges$ is sufficiently large. Setting $\delta' = \frac{\delta}{3},$ we have that
    \begin{equation}\label{eq:dragonfruit-1}
        \left(P_{\poly} - \EE\left[X^{(\numcolleges)}\right]\right) - v_\totalsupply < \frac{2\delta}{3}.
    \end{equation}
    when $\numcolleges$ is sufficiently large.

    We may analogously use (ii) to show that
    \begin{equation}\label{eq:dragonfruit-2}
        v_\totalsupply - \left(P_{\poly} - \EE\left[X^{(\numcolleges)}\right]\right) < \frac{2\delta}{3}.
    \end{equation}
    Combining \eqref{eq:dragonfruit-1} and \eqref{eq:dragonfruit-2} shows the result.
\end{proof}

\subsubsection{\Cref{thm:wisdom}(ii)}
\begin{proof}[Proof of \Cref{thm:wisdom}(ii)]
We first show that $P_{\mono}$ is fixed for all $\numcolleges$. Indeed,
\begin{align}
    S &= \int_{-\infty}^\infty \Pr[\mu_{\mono}(v)\in \colleges]\,d\eta(v)\\
    &= \int_{-\infty}^\infty \Pr[v + X > P_{\mono}]\,d\eta(v).\label{eq:red}
\end{align}
It follows from the proof of \Cref{prop:equal-cutoffs} that there is a unique solution to \eqref{eq:red}. Therefore, $P_{\mono}$ is the same for all $\numcolleges$. It follows directly that $\Pr[v + X > P_{\mono}]$ is constant in $\numcolleges$. 

It remains to show that college welfare is suboptimal for $\mu_{\mono}$---i.e., that
\begin{equation}
    \int_{-\infty}^\infty v \Pr[v + X > P_{\mono}]\,d\eta(v) < \int_{v_\totalsupply}^\infty v\,d\eta(v).
\end{equation}
Since $\int_{-\infty}^\infty \Pr[v + X > P_{\mono}]\,d\eta(v) = \totalsupply = \int_{v_\totalsupply}^\infty v\,d\eta(v),$ it suffices to show that $\Pr[v+X > P_{\mono}]\in (0,1)$ on a set $I$ of positive $\eta$-measure. The inequality clearly holds for $I = (P_{\mono} - X_+, P_{\mono} - X_-)$, where we use that $X$ has connected support $[X_-, X_+]$. (Specifically, \Cref{prop:cdf-increasing} gives that $F_X(x)$ is strictly increasing and lies in $(0,1)$ for $x\in (X_-, X_+)$.) Finally, we have from \eqref{eq:interval-mono} that $I$ has positive measure.
\end{proof}

\subsection{Proofs for \Cref{sec:results-student}}

\begin{proof}[Proof of \Cref{thm:top-choice}]
We show the three parts separately:
\\
\\
\textbf{Part (i).} To show part (i), note that 
\begin{equation}
    \Pr[\Rank_v(\mu_{\mono}(v)) = 1] = \Pr[v + X > P_{\mono}].
\end{equation}
This can be contrasted with the polyculture setting, in which
\begin{equation}
    \Pr[\Rank_v(\mu_{\poly}(v)) = 1] = \Pr[v + X > P_{\poly}].
\end{equation}
By \Cref{cor:cutoff-inequality}, $P_{\mono} < P_{\poly}$ for all $\numcolleges$, from which the non-strict inequality. To characterize when the inequality is strict, note that
\begin{align}
    \Pr[\Rank_v(\mu_{\mono}(v)) = 1] &= \Pr[v + X > P_{\mono}] = 1 - F_X(P_{\mono} - v)\\
    \Pr[\Rank_v(\mu_{\poly}(v)) = 1] &= \Pr[v + X > P_{\poly}] = 1 - F_X(P_{\poly} - v).
\end{align}
When $v\in (P_{\mono} - X_+, P_{\mono} - X_-),$ we have that
\begin{equation}
    P_{\mono} - v \in (X_-, X_+).
\end{equation}
Noting that $P_{\poly} - v > P_{\mono} - v,$ the result follows from \Cref{prop:cdf-increasing}, which shows that $F_X$ is strictly increasing on $(X_-, X_+).$ Also, we have already shown that this interval is of positive measure \eqref{eq:interval-mono}.
\\
\\
\textbf{Part (ii).} Part (ii) is a direct consequence of the Equal Cutoffs Lemma (\Cref{prop:equal-cutoffs}), since if, under monoculture, the student's score is above the cutoff for one college, then it is above the cutoff for all colleges. Therefore, any student who receives an offer from one college receives an offer from all colleges, and matches with their top choice.
\\
\\
\textbf{Part (iii).} 
We first show that $v_\totalsupply < P_{\mono} - X_-$. Indeed,
\begin{align}
    \int_{P_{\mono} - X_-}^\infty 1 \,d\eta(v) &= \int_{P_{\mono} - X_-}^\infty \Pr[\mu_{\mono}(v)\in \colleges] \,d\eta(v)\\
    &< S\\
    &= \int_{v_\totalsupply}^\infty 1\,d\eta(v).
\end{align}
It is also clear that $v_\totalsupply < V_+$.
Therefore, $(v_\totalsupply, P_{\mono} - X_-)$ is a non-empty interval that intersects $(V_-, V_+),$ and so has positive $\eta$-measure from \Cref{prop:nonzero-measure}. It suffices to show that $\Pr[\mu_{\mono}(v)\in \colleges] < \Pr[\mu_{\poly}(v)\in \colleges]$ on this interval.

When $v < P_{\mono} - X_-,$
    \begin{equation}
        \Pr[\mu_{\mono}(v)\in \colleges] = 1 - F_X(P_{\mono} - v) < 1.
    \end{equation}
    Note here that $1 - F_X(P_{\mono} - v)$ is constant in $\numcolleges$, so for all $v\in (v_\totalsupply, P_{\mono} - X_-),$
    \begin{equation}
        \lim_{C\rightarrow \infty} \Pr[\mu_{\mono}(v)\in \colleges] < 1.
    \end{equation}
    Meanwhile, for any $v > v_\totalsupply$, \Cref{thm:wisdom} gives that
    \begin{equation}
        \lim_{C\rightarrow \infty} \Pr[\mu_{\poly}(v)\in \colleges] = 1.
    \end{equation}
    It follows that for all $(v_\totalsupply, P_{\mono} - X_-)$, for $\numcolleges$ sufficiently large,
    \begin{equation}
        \Pr[\mu_{\poly}(v)\in \colleges] > \Pr[\mu_{\mono}(v)\in \colleges].
    \end{equation}
\end{proof}

\section{Proofs for \Cref{sec:diff-access}}

\begin{proof}[\textbf{Proof of \Cref{lem:equal-cutoffs-diff-access}}]
Let $Z$ be the set of market-clearing cutoffs. Then by symmetry, if $P = (P_1, \cdots, P_\numcolleges) \in Z$, then any permutation of $P$ is also in $Z$. It follows that
    \begin{equation}
        \sup_{P\in Z} P_1 = \sup_{P\in Z} P_2 = \cdots = \sup_{P\in Z} P_\numcolleges
    \end{equation}
    and
    \begin{equation}
        \inf_{P\in Z} P_1 = \inf_{P\in Z} P_2 = \cdots = \inf_{P\in Z} P_\numcolleges.
    \end{equation}
    Call these two common values $P_+$ and $P_-$ respectively. By \Cref{prop:lattice}, $(P_+, \cdots, P_+), (P_-, \cdots, P_-)\in Z$. These are the greatest and least elements of the complete lattice of market-clearing cutoffs. Therefore, it suffices to show that $P_+ = P_-.$
    
    We show the result for $\theta = \theta_{\poly, \kappa},$ and the result follows similarly for$\theta = \theta_{\mono, \kappa}.$ It suffices to show that $P_+ = P_-$. Assume for sake of contradiction that $P_+ > P_-$. (By definition, $P_+\ge P_-$.) By the definition of market clearing, the total measure of students matched is equal to $S$ under either $(P_-, \cdots, P_-)$ or $(P_+, \cdots, P_+)$. Therefore,
    \begin{equation}
        \int_\RR \Pr[v + X^{(\kappa)} > P_+]\,d\eta(v) = \totalsupply = \int_\RR \Pr[v + X^{(\kappa)} > P_-]\,d\eta(v).
    \end{equation}
    To give the desired contradiction, we show that
    \begin{equation}
        \int_\RR \Pr[v + X^{(\kappa)} > P_+]\,d\eta(v) < \int_\RR \Pr[v + X^{(\kappa)} > P_-]\,d\eta(v).
    \end{equation}
    Clearly, $\Pr[v + X^{(\kappa)} > P_+] \le \Pr[v + X^{(\kappa)} > P_-].$ Therefore, it suffices to show that $\Pr[v + X^{(\kappa)} > P_+] < \Pr[v + X^{(\kappa)} > P_-]$ on a set $I$ of positive $\eta$-measure. We show this is the case for $I = (P_- - X_+, P_- - X_-).$
    
    We first show that the inequality is satisfied on $I$. Note that for $v\in I,$ we have that $P_- - v \in (X_-, X_+).$ Therefore,
    \begin{align}
        \Pr[v + X^{(\kappa)} > P_-] - \Pr[v + X^{(\kappa)} > P_+] = \sum_{k=1}^\numcolleges \Pr[\kappa = k] \left(F_X^k(P_+ - v) - F_X^k(P_- - v)\right) > 0
    \end{align}
    where the last inequality follows from observing that by $P_+-v > P_- - v$ by assumption and applying \Cref{prop:cdf-increasing}, which says that $F_X^k$ is strictly increasing on $(X_-, X_+)$.

    It remains to show that the interval $I$ has positive $\eta$-measure. We begin by showing that
    $P_- > V_- + X_-.$
    This is true because for $P \le V_- + X_-$,
    \begin{equation}
        \int_\RR \Pr[v + X^{(\kappa)} > P]\,d\eta(v) \ge \int_\RR \Pr[v + X > V_- + X_-]\,d\eta(v) = \int_\RR 1\,d\eta(v) = 1,
    \end{equation}
    contradicting the assumption that $S < 1$.
    Similarly, $P_- < V_+ + X_+.$ It follows that $P_- - X_+ < V_+$ and $P_- - X_- > V_-.$
    Therefore, the interval $I = (P_- - X_+, P_- - X_-)$ intersects $(V_-, V_+)$. So \Cref{prop:nonzero-measure} implies that it has positive measure.
\end{proof}

\begin{proof}[\textbf{Proof of \Cref{thm:diff-app-access}}]
    To show part (i), note that $P_{\mono} = P_{\mono, \kappa}$.
    
    To show part (ii), note that
    \begin{align}
        \Pr[\mu_{\poly, \kappa}(v)\in \colleges\,|\,k(v)=k]
        &= \Pr[v + X^{(k)} > P_{\poly, \kappa}]\\
        &= 1 - F_X^k(P_{\poly, \kappa} -v),\label{eq:clock}
    \end{align}
    which is increasing in $k$. Moreover, when $v\in (P_{\poly, \kappa} - X_+, P_{\poly, \kappa} - X_-),$ we have that $P_{\poly, \kappa} - v\in (X_-, X_+),$ so by \Cref{prop:cdf-increasing}, $F_X(P_{\poly, \kappa})\in (0,1),$ in which case \eqref{eq:clock} is strictly increasing in $k$.
\end{proof}

\end{document}